\documentclass[journal]{IEEEtran}

\usepackage{cite}
\usepackage{amsbsy,amsmath,amssymb,epsfig,bbm,mathrsfs,multirow,amsthm,amsfonts}
\usepackage{algorithm,algorithmic}
\usepackage{graphicx}
\usepackage{textcomp}
\usepackage{multicol}
\usepackage{url}
\usepackage{ragged2e}

\usepackage[table]{xcolor}

\PassOptionsToPackage{hyphens}{url}
\usepackage{hyperref}
\hypersetup{colorlinks, citecolor=green, filecolor=pink, linkcolor=blue, urlcolor=blue }
\graphicspath{{./figures/}}

\def\BibTeX{{\rm B\kern-.05em{\sc i\kern-.025em b}\kern-.08em
    T\kern-.1667em\lower.7ex\hbox{E}\kern-.125emX}}

\newtheorem{theorem}{Theorem}

\usepackage[T1]{fontenc}
\usepackage{makecell}
\usepackage{amstext} 
\usepackage[labelformat=simple]{subcaption}

\definecolor{mygray}{gray}{0.6}
\definecolor{myblue}{rgb}{0.8,0.85,1}

\usepackage{array, tabularx, boldline}
\newcolumntype{L}[1]{>{\raggedright\let\newline\\\arraybackslash\hspace{0pt}}m{#1}}
\newcolumntype{C}[1]{>{\centering\let\newline\\\arraybackslash\hspace{0pt}}m{#1}}
\newcolumntype{R}[1]{>{\raggedleft\let\newline\\\arraybackslash\hspace{0pt}}m{#1}}

\usepackage{cellspace}
\begin{document}

\title{Edge Association Strategies for Synthetic Data Empowered Hierarchical Federated Learning with Non-IID Data}

\author{
Jer Shyuan Ng\thanks{J.~S.~Ng, A.~P.~Kalapaaking, X.~Xia, I.~Khalil and I.~Gondal are with Royal Melbourne Institute of Technology, Australia. (Email: jer.shyuan.ng@rmit.edu.au.)}, 
Aditya Pribadi Kalapaaking, \\
Xiaoyu Xia, 
Dusit Niyato\thanks{D.~Niyato is with School of Computer Science and Engineering, Nanyang Technological University, Singapore.}~\textit{IEEE~Fellow}, 
Ibrahim Khalil, 
Iqbal Gondal
}
\maketitle

\begin{abstract}
In recent years, Federated Learning (FL) has emerged as a widely adopted privacy-preserving distributed training approach, attracting significant interest from both academia and industry. Research efforts have been dedicated to improving different aspects of FL, such as algorithm improvement, resource allocation, and client selection, to enable its deployment in distributed edge networks for practical applications. One of the reasons for the poor FL model performance is due to the worker dropout during training as the FL server may be located far away from the FL workers. To address this issue, an Hierarchical Federated Learning (HFL) framework has been introduced, incorporating an additional layer of edge servers to relay communication between the FL server and workers. While the HFL framework improves the communication between the FL server and workers, large number of communication rounds may still be required for model convergence, particularly when FL workers have non-independent and identically distributed (non-IID) data. Moreover, the FL workers are assumed to fully cooperate in the FL training process, which may not always be true in practical situations. To overcome these challenges, we propose a synthetic-data-empowered HFL framework that mitigates the statistical issues arising from non-IID local datasets while also incentivizing FL worker participation. In our proposed framework, the edge servers reward the FL workers in their clusters for facilitating the FL training process. To improve the performance of the FL model given the non-IID local datasets of the FL workers, the edge servers generate and distribute synthetic datasets to FL workers within their clusters. FL workers determine which edge server to associate with, considering the computational resources required to train on both their local datasets and the synthetic datasets. The simulation results show that an evolutionary equilibrium is reached where the FL workers do not have incentive to change their edge association strategies. Given this equilibrium, the FL workers facilitate the FL training of the edge servers that they associate with and be rewarded for their contributions. The proposed framework achieves higher FL model accuracy with an addition of 5\% of synthetic data.
\end{abstract}


\section{Introduction}
\label{sec:intro}

In recent years, Artificial Intelligence (AI) has achieved remarkable progress across various industries, including transportation, healthcare, manufacturing, and finance. Advancements in hardware, improvements in AI algorithms, and the exponential growth of data have enabled AI models to identify patterns and make decisions that closely mimic human behavior. For instance, autonomous mobile robots can navigate warehouses independently, accurately placing items on designated shelves without human intervention. Additionally, AI-powered models assist in validating medical diagnoses and predicting diseases by analyzing patient imagery~\cite{endargiri2020automatic}.

Traditionally, AI models are trained on centralized servers, requiring local devices to transfer their data to a central server. However, with increasingly stringent data privacy regulations such as the General Data Protection Regulation (GDPR) and the rising demand for communication resources to handle large-scale data transmission, FL has emerged as a key enabler for AI model training in distributed edge computing networks. FL enables collaborative learning without transferring raw user data, ensuring privacy and reducing communication overhead. FL models have been applied in various domains, including vehicle target recognition~\cite{jallepalli021federated}, traffic flow prediction~\cite{liu2020privacy}, and enhancing agricultural productivity~\cite{friha2022felids}.

Future FL networks are expected to consist of thousands of heterogeneous mobile and Internet-of-Things (IoT) devices. However, one of the main challenges of FL is communication inefficiency, especially when the FL server is physically distant from FL workers. The effectiveness of FL training is often hampered by a high dropout rate among participating devices due to limited communication resources, making it difficult for workers to transmit updated local model parameters over long distances~\cite{mcmahan2017communication}. To address this issue, the HFL framework has been introduced~\cite{lim2022decentralized, lim2021dynamic, ng2022reputation, xu2022adaptive}. The three-layer structure of HFL combines the advantages of both cloud-based and edge-based FL~\cite{liu2020client}. On one hand, HFL reduces communication costs by leveraging edge servers, which are closer to FL workers, similar to edge-based FL. On the other hand, it has a larger pool of training data by allowing large number of FL workers to participate in the FL training, similar to cloud-based FL.

While HFL effectively minimizes training latency by reducing node failures and dropout rates, several challenges remain that hinder its large-scale deployment in distributed edge networks. Firstly, FL training in HFL networks lacks sufficient incentives for edge devices to participate~\cite{pan2025bilateral}. The FL server, edge servers, and edge devices may be operated by different entities, each seeking to maximize its own utility while operating under resource constraints. As a result, edge devices may not be willing to collaborate in the FL training. Secondly, FL model convergence may be slow as the data distributions of the FL workers may vary significantly. This is because the edge devices collect data independently based on their unique preferences and sampling spaces. Furthermore, differences in sampling rates can lead to imbalanced datasets, where some edge devices possess significantly larger datasets than others. As a result, the FL model performance may degrade due to the statistical heterogeneity of local datasets~\cite{lu2024noniid, li2024iofl}.  These challenges highlight the need to address the statistical heterogeneity of non-independent and identically distributed (non-IID) data while also incentivizing edge device participation in FL training tasks.

 \begin{figure}
\centering
\includegraphics[width=\linewidth]{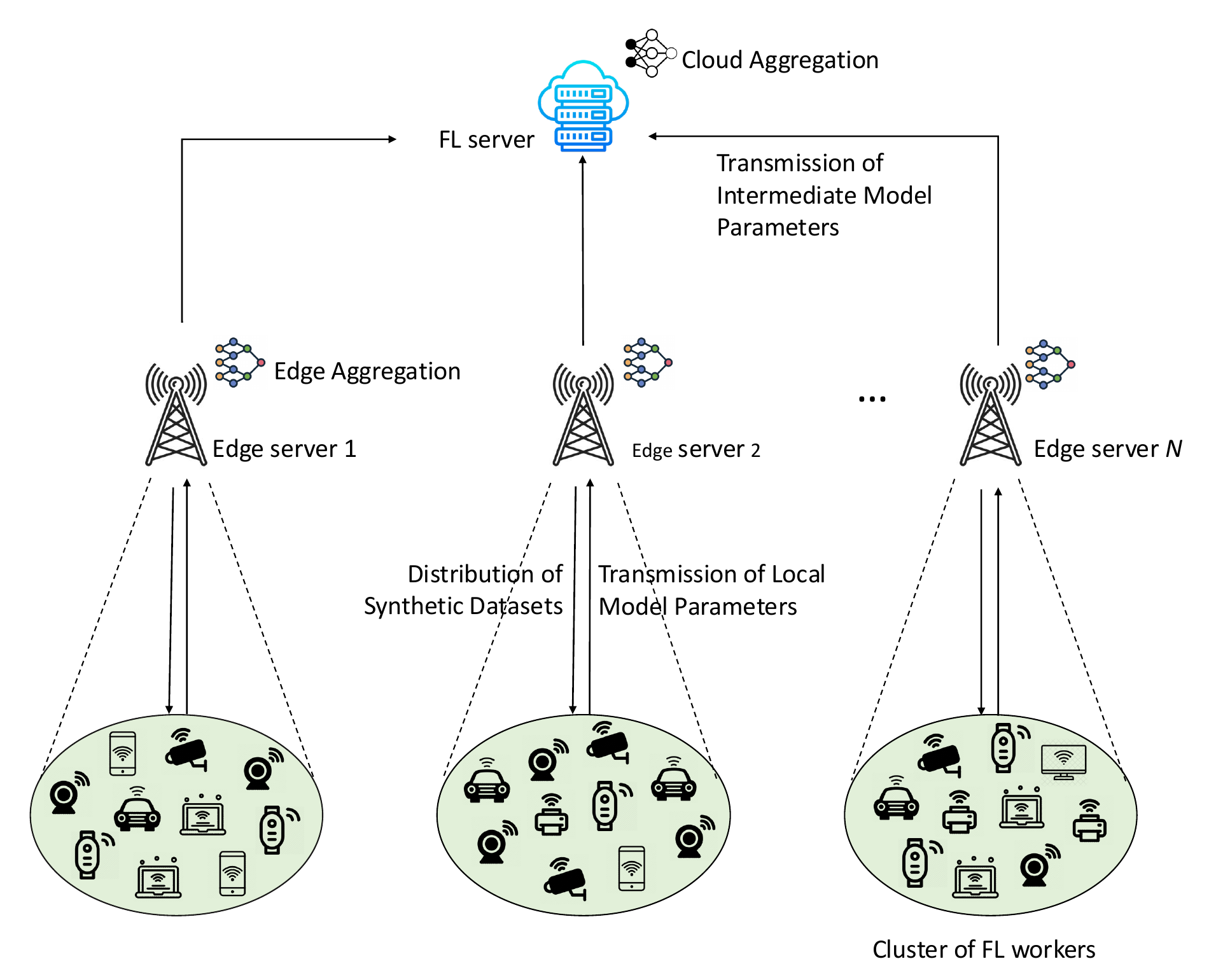}
\caption{Synthetic-data-empowered Hierarchical Federated Learning where the edge servers distribute synthetic datasets to the FL workers}
\label{fig:system}
\end{figure}

In this paper, we explore an HFL network consisting of an FL server that aims to train an FL model, such as a route recommendation system. To facilitate the training process, the FL server engages multiple edge servers (e.g., base stations), which in turn recruit nearby IoT devices as FL workers. These IoT devices transmit their locally trained model parameters to the edge servers, where intermediate model aggregation is performed. The edge servers then relay the aggregated intermediate parameters to the FL server, which conducts the final global model aggregation.

To encourage participation in the FL training process, each edge server offers a reward pool that is distributed among the FL workers in its cluster based on their contributions. Since this reward pool remains fixed regardless of the number of participants, as more FL workers join a particular cluster, the reward per worker decreases. Consequently, FL workers must strategically select an edge server that maximizes their individual utility. Over time, as workers evaluate their rewards relative to others, some may decide to switch to different edge servers to optimize their benefits. This dynamic process continues until an equilibrium is reached, meaning no worker can further increase its utility by switching edge servers. To model this evolving edge association behavior, we employ an evolutionary game framework, allowing us to determine the equilibrium composition and cluster formation of the edge servers.

A major challenge in HFL is the non-IID nature of the data collected by FL workers, as each device gathers data personalized to user behaviors and environmental conditions. This statistical heterogeneity often leads to slow model convergence and reduced performance. To address this issue, edge servers generate task-specific synthetic datasets using pre-trained AI models and distribute them to FL workers. As a result, FL workers train the FL model using both their local datasets and the synthetic datasets provided by the edge servers.

The primary goal of this work is to address the statistical challenge of non-IID data while simultaneously encouraging FL worker participation, thereby enabling the practical deployment of FL models in an HFL network. The key contributions of this work are  described as follows:

\begin{enumerate}

\item We introduce a synthetic-data-empowered HFL framework (Fig.~\ref{fig:system}) that integrates synthetic data to address the issue of non-IID data distributions among FL workers, improving model performance and convergence speed. 

\item We model the decision-making process of FL workers as an evolutionary game, where workers strategically choose edge servers based on the available rewards and required resource allocation. Given that FL workers must balance resources between training on their local data and synthetic datasets, the proposed model effectively captures their adaptive behaviors in edge server selection.

\item We conduct simulations to assess the performance of the synthetic-data-empowered HFL network. Our results demonstrate that evolutionary equilibrium is unique and stable, achieving that FL workers ultimately settle into optimal edge associations. Furthermore, the introduction of synthetic datasets significantly enhances FL model accuracy, reducing the number of training iterations needed to achieve desired performance levels under non-IID data conditions. The simulation results show that the FL model accuracy can be increased by just adding 5\% of synthetic data under the most extreme non-IID condition where each FL worker has only a single class of data samples.

\end{enumerate}

The rest of the article is organized as follows. Section~\ref{sec:related} reviews related work. Section~\ref{sec:system} presents the system model. Sections~\ref{sec:evolution} and~\ref{sec:evaluation} discuss the evolutionary game for cluster formation and the simulation results respectively. Section~\ref{sec:conclusion} concludes the paper.

\section{Related Work}
\label{sec:related}

FL is a privacy-preserving distributed machine learning approach that enables the development of personalized models without requiring users to share their sensitive personal data with a centralized server. However, in an FL network consisting of numerous distributed devices located far from the FL server, high communication overhead remains a critical challenge~\cite{mcmahan2017communication}. To tackle this, researchers have explored various communication-efficient FL strategies, including distributed network optimization methods such as Delayed Gradient Averaging~\cite{zhu2021delayed} to reduce network latency, model compression~\cite{lin2020deep} and quantization~\cite{alistarh2017qsgd} techniques to minimize bandwidth usage, user selection protocols~\cite{chen2021joint} to mitigate the straggler effects, resource allocation schemes to improve wireless network efficiency~\cite{chen2022federated, cao2022resource} and robust designs to mitigate the impact of communication noise during FL training~\cite{ang2020robust}.

With advancements in computing and communication capabilities, edge nodes (e.g., IoT devices and edge servers) are increasingly leveraged for FL over mobile edge networks~\cite{lim2020federated}. In the study of~\cite{gouissem2022robust}, the authors have proposed a serverless and robust FL mechanism where the edge nodes train a neural network collaboratively while being resilient to Byzantine attacks. This decentralized approach mitigates the issue of single point of failure and is fully dependent on the participating edge nodes to make collaborative decisions in isolating malicious nodes. In addition to reducing the reliance on the central server, the HFL framework introduced in~\cite{liu2020client} also substantially lowers communication costs by decreasing the number of communication rounds between the FL server and edge devices. Specifically, edge servers serve as an intermediate layer, performing local model aggregation before forwarding the aggregated parameters to the FL server for global aggregation.

In real-world mobile edge networks, heterogeneity is an inherent characteristic of IoT environments. For example, the edge nodes are heterogeneous in terms of their resources, e.g., bandwidth, storage capacity and computation power. In consideration of resource heterogeneity of the edge devices, resource-efficient FL schemes~\cite{lim2022decentralized, lim2021dynamic, wang2023accelerating, jung2022resource} are proposed for the edge servers to form clusters with the edge devices in the HFL network. In the work of~\cite{wang2023accelerating}, the edge devices are grouped into different clusters based on their heterogeneous training capacities. In the lower level, the edge devices transmit their local model parameters to their respective edge servers for intermediate model aggregation in a synchronous manner. In the upper level, the edge servers transmit the intermediate model parameters to the FL server for global aggregation asynchronously. Apart from device heterogeneity, statistical heterogeneity in FL arises due to the diverse data distributions of edge devices. The non-IID data distribution of the edge nodes may lead to poor performance, slow convergence and inconsistent gradient updates. In addition to the latency, energy and bandwidth consumption constraints of the edge devices, the data distributions of edge devices are also considered in the optimization formulations of the resource allocation schemes~\cite{alaa2022communication}.  Other approaches to mitigate the statistical challenge in the FL network include adopting sampling methods on the local training datasets~\cite{li2021sample, zhu2023isfl} and utilizing a learned auxiliary local drift variable to bridge the gap between the local and global model parameters~\cite{gao2022feddc}. To address the heterogeneity challenge, i.e., device, model and statistical heterogeneities, in the IoT environment, the authors in~\cite{wu2020personalized} introduced a personalized FL framework within a cloud-edge architecture, enabling each edge node to develop its own customized FL model tailored to its specific resource constraints and application needs.

Data sharing is another approach used to improve FL model performance in non-IID settings. The authors in~\cite{yoshida2020hybrid} have proposed an Hybrid-FL protocol that requires the selected edge devices to upload their data to the FL server to construct an approximate IID dataset. However, by uploading their data to the mobile edge computing operator, the FL participants may suffer from unexpected information leakage to attackers. Sensitive information, e.g., age, gender, nationality and address, can be inferred by malicious participants just by inspecting the shared FL model~\cite{melis2019exploiting}, violating the privacy-preserving principle of FL. In~\cite{zhao2018federated}, the authors have shown that the accuracy of FL models has decreased by up to 55\% under highly-skewed non-IID data distributions. To this end, they have proposed to create a small subset of data that is globally shared between all edge devices. The FL model accuracy can be increased by 30\% with only 5\% of globally shared data on the CIFAR-10 dataset. Instead of relying on globally shared data that is not always available, the work in~\cite{li2024feature} propose a hard feature matching data synthesis method for FL where each FL worker generates hard-feature-matching synthetic dataset that is shared with other FL workers. Inspired by these approaches, our proposed HFL framework incorporates synthetic data generation at the edge servers, which is then distributed to FL workers for training.

Beyond statistical challenges, FL worker participation incentives remain a significant concern. FL workers are typically unmotivated to contribute local data and may be reluctant to allocate additional computational resources to train models with synthetic datasets. In our framework, we address this issue by integrating an evolutionary game approach to optimize resource allocation and encourage participation, ensuring a fair and efficient FL training process across the HFL network.

\section{System Model}
\label{sec:system}

\begin{table}[t]
\scriptsize
  \caption{System Model Parameters.}
  \label{tab:parameters}
  \centering
  \begin{tabularx}{8.7cm}{|Sl|X|}
  \hline
  \rowcolor{mygray}
 \textbf{Parameter} & \textbf{Description} \\   \hline
$c_z$ & Computation resources for local datasets\\ \hline
$d_z$ & Amount of data quantities\\ \hline
$m_z$ & Commmunication resources\\ \hline
$s_n$ & Computation resources for synthetic datasets\\ \hline
$x_n^{(z)}$ & Population share\\ \hline
$\boldsymbol{x}^{(z)}$ & Population state\\ \hline
$J$ & Number of FL workers\\ \hline
$K$ & Number of training iterations\\ \hline
$N$ & Number of edge servers \\ \hline
$Z$ & Number of populations\\ \hline
$\alpha$ & Unit computation cost\\ \hline
$\beta$ & Unit communication cost\\ \hline
$\eta_k$ & Gradient descent step size\\ \hline
$\kappa_1$ & Number of local updates\\ \hline
$\kappa_2$ & Number of intermediate model aggregations\\ \hline
$\boldsymbol{\omega}_j^n$ & Local model parameters\\ \hline
$\boldsymbol{\omega}^n$ & Intermediate model parameters\\ \hline
$\boldsymbol{\omega}$ & Global model parameters\\ \hline
\end{tabularx}
\end{table}

We consider a distributed edge computing network comprising three key entities (i) an FL server, i.e., model owner, (ii) $J$ FL workers, i.e., IoT devices, the set of which is represented by $\mathcal{J}=\{1,\ldots, j, \ldots, J\}$ and $N$ edge servers, e.g., base stations, the set of which is represented by $\mathcal{N}=\{1,\ldots, n, \ldots, N\}$. The FL server trains a model, with parameters denoted by a vector $\boldsymbol{\omega}$, over $K$ iterations represented by the set $\mathcal{K} = \{1, \ldots, k, \ldots, K\}$ to minimize global loss $L^{K}(\boldsymbol{\omega}$). This loss function is defined by the FL server based on the specific learning task being executed. The training process is performed based on an HFL framework, which introduces an intermediate aggregation layer between the FL server and FL workers.

Unlike conventional FL, the HFL framework incorporates edge servers as an intermediate layer to aggregate local model updates before transmitting them to the FL server. This hierarchical structure enhances communication efficiency by reducing latency without compromising model accuracy~\cite{abad2020hierarchical}. Each edge server $n$ forms a cluster consisting of several FL workers.
The FL server first distributes the global FL model parameters $\boldsymbol{\omega}{(k)}$ to the edge servers. Each edge server $n$ in turn forwards these global FL model parameters to the FL workers in its cluster. Each FL worker $j$ that is associated with edge server $n$ trains the global FL model with its local dataset and derives the updated local model parameters $\boldsymbol{\omega}_j^n{(k+1)}$. Once local training is completed, the FL workers ($\forall j \in \mathcal{J}$) transmit the updated local model parameters $\boldsymbol{\omega}_j^n{(k+1)}$ to their respective edge servers. After every $\kappa_1$ training iterations, the edge servers perform intermediate aggregation to derive the updated intermediate model parameters $\boldsymbol{\omega}^n{(k+\kappa_1)}$. The updated intermediate model parameters are then redistributed to the FL workers for continued training. After every $\kappa_2$ intermediate aggregations, all edge servers ($\forall n \in \mathcal{N}$) transmit the updated intermediate model parameters $\boldsymbol{\omega}^n{(k+\kappa_1\kappa_2)}$ to the FL server for global aggregation to derive the updated global model parameters $\boldsymbol{\omega}{(k+\kappa_1\kappa_2)}$. The total training process consists of $T$ cloud intervals, each spanning $\kappa_1\kappa_2$iterations, meaning that the FL server only performs cloud aggregation every $\kappa_1\kappa_2$ training iterations instead of every iteration, significantly reducing communication overhead. The evolution of the local model parameters of FL worker $j$ that is associated with edge server $n$ at training iteration $k$, $\boldsymbol{\omega}_j^{n}(k)$, is defined as:

\begin{equation}
\label{eqn:parameters}
\boldsymbol{\omega}_j^{n}(k) = 
\begin{cases}
\boldsymbol{\omega}_j^{n}(k-1)-\eta_k\nabla F_j^n(\boldsymbol{\omega}_j^{n}(k-1)), & k|\kappa_1 \neq 0,\\[3ex]
\frac{\sum_{j \in }|D_j^n|\left[\boldsymbol{\omega}_j^{n}(k-1)-\eta_k\nabla F_j^n(\boldsymbol{\omega}_j^{n}(k-1))\right]}{|D^n|},  &k|\kappa_1=0,\\ 
& k|\kappa_1\kappa_2\neq 0,\\[3ex]
\frac{\sum_{j=1}^{J}|D_j^n|\left[\boldsymbol{\omega}_j^{n}(k-1)-\eta_k\nabla F_j^n(\boldsymbol{\omega}_j^{n}(k-1))\right]}{|D|}, & k|\kappa_1\kappa_2=0,
\end{cases}
\end{equation}
where $F_j^n$ is the local loss function of FL worker $j$ that is associated with edge server $n$ and$\eta_k$ is the gradient descent step size at iteration $k$, . $D_j^n$ is the local dataset of FL worker $j$ that is associated with edge server $n$, $D^n$ is the aggregated dataset of edge server $n$ and $D$ is the total dataset where $\bigcup_{j=1}^J D_j^n = D$, $\forall n \in \mathcal{N}$. The first, second and third caseas of Equation~(\ref{eqn:parameters}) corresponds to the local model parameters of FL worker $j$ when no aggregation, intermediate aggregation and global aggregation take place, respectively.



While the HFL framework enables FL model training without exposing the private data of FL workers, it still encounters several key challenges. Firstly, there is lack of incentives for FL workers. The FL workers may be reluctant to contribute their resources, e.g., data, computational power, and communication bandwidth, toward training the FL server’s model since there is no direct benefit or reward for their participation. Secondly, due to the naturally non-IID nature of data among FL workers, the training process can experience high latency. This statistical heterogeneity presents challenges from two key perspectives: (i) the divergence of local models, and (ii) data imbalance and overfitting. From the first perspective, since different FL workers have varying interests and usage patterns, their data distributions differ significantly. Although this setting mirrors real-world data distributions, it leads to substantial variations in local models. From the second perspective, some FL workers may have much smaller datasets than others, leading to an imbalance that can cause local models to overfit to the limited data available. These factors contribute to slow convergence and degraded FL model performance~\cite{hsieh202noniid}. As a result, while the HFL framework effectively reduces communication latency, this advantage can be negated by the increased latency arising from statistical heterogeneity in the data distribution.

In this paper, we consider that the $N$ edge servers are pre-selected by the FL server using selection algorithms based on various criteria, such as social factors~\cite{sun2020social}, resource efficiency~\cite{wu2023energy} or reputation~\cite{xin2022node}. To mitigate the statistical challenge of the HFL framework, we propose a synthetic-data-empowered HFL framework. Specifically, within each cloud interval $T$, the following steps are performed:


\begin{enumerate}

\item \emph{Distribution of Global Model: }The FL server shares the global FL model parameters denoted as $\boldsymbol{\omega}{(k)}$ with the edge servers. Each edge server $n$ in turn forwards the global FL model parameters to the FL workers. For clarity, we refer to the model parameters derived by the training of FL model on each FL worker $j$'s local dataset as the \emph{local} model parameters, the model parameters aggregated across all FL workers as the \emph{intermediate} model parameters and the model parameters aggregated across all edge servers as the \emph{global} model parameters.

\item \emph{Distribution of Synthetic Datasets: }The edge servers generate synthetic datasets related to the FL training task using pre-trained AI models. For instance, for a number classification task, a synthetic dataset equivalent to the \emph{MNIST} dataset is created by using the conditional Generative Adversarial Network (cGAN) introduced in~\cite{rahmdel2024cgan}. The cGAN model comprises a generator model that generates synthetic data and a discriminator model that distinguishes between synthetic and real data samples. For an object classification task, \emph{CIFAKE} dataset~\cite{bird2024cifake}, that is a synthetic equivalent to the \emph{CIFAR-10} dataset, is used. The \emph{CIFAKE} dataset is generated by using Latent Diffusion Models (LDMs) that model the diffusion of image data through a latent space given a textual context. Then, the edge servers ($\forall n \in \mathcal{N}$) transmit these synthetic datasets to the FL workers within their respective clusters.

\item \emph{FL Local Training: }Each FL worker $j$ trains the global FL model using its local dataset and the synthetic dataset received from its associated edge server to derive the updated local model parameters $\boldsymbol{\omega}_j^{n}{(k+1)}$. 

\item \emph{Wireless Transmission of Local Model Parameters: }The FL worker ($j \in \mathcal{J}$) transmits the updated local parameters $\boldsymbol{\omega}_j^{n}{(k+1)}$ to its corresponding edge server upon completion of its local FL model training.

\item \emph{Aggregation of Intermediate Model Parameters: }After every $\kappa_1$ training iterations, the edge server $n$ aggregates the local model parameters from all FL workers in its cluster to derive the updated intermediate model parameters $\boldsymbol{\omega}^n{(k+\kappa_1)}$. This updated intermediate model parameters are transmitted back to the FL workers for their $(k+\kappa_1)^{th}$ training iteration.

\item \emph{Aggregation of Global Model Parameters: } At the end of predefined intervals $\kappa_1\kappa_2$, all edge servers transmit the updated intermediate model parameters to the FL server to derive the updated global model parameters $\boldsymbol{\omega}{(k+\kappa_1\kappa_2)}$. This updated global model is used to conduct the next $T+1$ cloud interval of training with the edge servers and their associated FL workers.
\end{enumerate}

To enhance the quality of their intermediate models, edge servers seek to attract more FL workers to contribute valuable resources such as computation, communication, and data. To incentivize participation, each edge server provides a reward pool to the FL workers associated it. The amount of reward each FL worker receives is determined by their contributions. The more resources an FL worker provides, the larger its share of the reward pool.

We assume that each FL worker can only support the FL training of a single edge server due to limited computational resources, making it infeasible for them to participate in multiple training tasks or perform parallel computations. As a result, each FL worker must decide which edge server to associate with. While workers are inclined to select edge servers offering larger reward pools, an increase in the number of associated workers means the rewards must be divided among more participants, thus lowering the individual payoff. This interdependence among FL workers' decisions creates a dynamic environment. To capture and analyze this behavior and the resulting edge server selection strategies, we employ an evolutionary game framework. The analysis and solution of this evolutionary game are presented in Section~\ref{sec:evolution}.

\section{Evolutionary Edge Association of FL Workers}
\label{sec:evolution}

In this section, we discuss the formulation of the evolutionary game, analyze the evolutionary game and present the algorithm of the evolutionary game.

\subsection{Evolutionary Game Problem Formulation}

In the evolutionary games, the association of FL workers to different edge servers is formulated as follows:

\begin{itemize}

\item \emph{Players: }The players in this evolutionary game are the $J$ FL workers, denoted as $\mathcal{J} = \{1, \ldots, j, \ldots, J\}$ and the $N$ edge servers, denoted as $\mathcal{N} = \{1, \ldots, n, \ldots, N\}$.

\item \emph{Population: }The FL workers are divided into $Z$ populations, denoted by $\mathcal{Z} = \{1, \ldots, z, \ldots, Z\}$. using clustering techniques such as k-means. We consider that the FL workers are partitioned into different populations based on their data quantities. FL workers within the same population have the same number of data samples, and the total data quantity in population $z$ is denoted as $d_z$.

\item \emph{Strategy: }Each FL worker $j \in \mathcal{J}$ aims to maximize its utility in Equation~(\ref{eqn:utility}).They must choose an edge server $n \in \mathcal{N}$ to associate with. Note that we only consider that each FL worker is associated with a single edge server and the FL worker contributes its data and resources for the FL training of the selected edge server.

\item \emph{Population share: }The population share, $x_n^{(z)}$, refers to the proportion of FL workers in population $z \in \mathcal{Z}$ that associate with edge server $n$, where $x_n^{(z)} \in [0,1]$ and the sum of shares over all edge servers is equal to one, i.e., $\sum_{n=1}^{N}x_n^{(z)}=1$, for each population $z$.

\item \emph{Population state: }The population state refers to the proportion of FL workers' association strategies. The population shares of all FL workers of population $z$ constitute the population distribution state of population $z$, which is represented by the vector $\boldsymbol{x}^{(z)} = [x_1^{(z)}, \ldots, x_n^{(z)},\ldots, x_N^{(z)}] \in \mathbb{X}^{(z)}$, where $\mathbb{X}^{(z)}$ represents the state space containing all possible population distributions of population $z$. 

\end{itemize}

When an FL worker in population $z$ associates with edge server $n$ and contributes data for its training, the utility for that worker is defined as:

\begin{equation}
\label{eqn:utility}
u_n^{(z)} = \mathcal{U}\left(\gamma_n \frac{d_z x_n^{(z)}}{\sum_{z=1}^{Z}d_z x_n^{(z)}} - \alpha (s_n + c_z) - \beta m_z \right),
\end{equation}
where $\gamma_n$ is the reward pool from edge server $n$, distributed among the associated FL workers based on their data contributions. The term $\frac{d_z x_n^{(z)}}{\sum_{z=1}^{Z}d_z x_n^{(z)}}$ represents the FL worker's share of the reward, which is proportional to the data quantity. The term can be easily extended to include different types of resources. For simplicity, we only consider FL workers' data quantities in the distribution of the reward pool in this paper. The term $c_z$ is the computation resource available to FL workers for training on local datasets, and $s_n$ is the additional computation resource required to train on the synthetic data from edge server $n$. The term $m_z$ is the amount of communication resource of FL workers in population $z$. The terms $\alpha$ and $\beta$ are the unit costs of computation and communication resources, respectively. The utility function $\mathcal{U}(\cdot)$ represents the risk attitude of the FL workers. Without loss of generality, we consider that the FL workers are risk-neutral, i.e., $\mathcal{U}(\cdot)$ is a linear utility function~\cite{marks2016risk}.

The FL worker’s objective is to maximize its utility:

\begin{equation}
\label{eqn:max}
\begin{gathered}
\max u_n^{(z)},\\
\text{s. t. }\\
\sum_{n=1}^{N}x_n^{(z)}=1, \forall z \in \mathcal{Z}\\
x_n^{(z)} = [0,1], \forall z \in \mathcal{Z}, \forall n \in \mathcal{N}.
\end{gathered}
\end{equation}

In the HFL network, the FL workers exchange information regarding the utilities they receive from different edge servers. Based on this information, they compare their utilities with those of other FL workers in the same population and may adjust their edge server association strategy in order to achieve higher utility. A change in the strategy of an FL worker may affect the utilities of other FL workers and thus, causing other FL workers to adjust their edge server association strategies. These dynamic interactions lead to the evolution of population shares and therefore population states over time. As such, the population state and population share of population $z$ are functions of time $t$, which can be denoted as $\boldsymbol{x}^{(z)}(t)$ and $x_n^{(z)}(t)$, respectively. Accordingly, the net utility of FL workers in population $z$ for their data and resources to support the FL training of edge server $n$ at time $(t)$ is:
\begin{equation}
\label{eqn:netutility}
u_n^{(z)}(t) = \mathcal{U}\left(\gamma_n \frac{d_z x_n^{(z)}(t)}{\sum_{z=1}^{Z}d_z x_n^{(z)}(t)} - \alpha (s_n + c_z) - \beta m_z \right).
\end{equation} 

To illustrate the dynamic evolutionary behaviour of the FL workers among different strategies, i.e., associating with different edge servers over time, we define the replicator dynamics, which are expressed as a set of ordinary differential equations, as follows:
\begin{equation}
\label{eqn:replicator}
\begin{split}
\dot{x}_n^{(z)}(t) = f_n^{(z)}(\boldsymbol{x}^{(z)}(t)) = \delta x_n^{(z)}(t)(u_n^{(z)}(t)-\bar{u}^{(z)}(t)),\\
\forall z \in \mathcal{Z}, \forall n \in \mathcal{N}, \forall t,
\end{split}
\end{equation}
with initial population state $\boldsymbol{x}^{(z)}(0) = \boldsymbol{x}^{(z)}_0 \in \mathbb{X}^{(z)}$. The rate of strategy adaption is represented by the term $\delta > 0 $, which controls the speed of which the FL workers adapt their strategies. The average utility of FL workers in population $z$, $\bar{u}^{(z)}(t)$, is defined as:
\begin{equation}
\label{eqn:average}
\bar{u}^{(z)}(t) = \sum_{n=1}^{N}u_n^{(z)}(t) x_n^{(z)}(t) .
\end{equation}
Based on the replicator dynamics, on the one hand, if the utility of an FL worker in population $z$ is more than the average utility received by all FL workers in the population, i.e., $u_n^{(z)}(t)>\bar{u}^{(z)}(t)$, the population share $x_n^{(z)}(t)$ increases as more FL workers choose to associate with this particular edge server $n$, i.e., $\dot{x}_n^{(z)}(t)>0$. On the other hand, if the utility of an FL worker in population $z$ is less than the average utility received by all FL workers in the population, i.e., $u_n^{(z)}(t)<\bar{u}^{(z)}(t)$, the FL worker may adjust its strategy and switch to another edge server.

At equilibrium, when the evolution stabilizes, the population shares no longer change, meaning $\dot{x}_n^{(z)}(t) = 0$, $\forall n \in \mathcal{N}$. At this equilibrium point, all FL workers in the same population receive the same utility. Therefore, there is no motivation for any FL worker to deviate from its current strategy by associating itself with another edge server. This is the evolutionary equilibrium, i.e., $\boldsymbol{x}^{(z)*} = [x^{(z)*}_1, \ldots, x^{(z)*}_n, \ldots, x^{(z)*}_N]$.

\subsection{Evolutionary Game Analysis}

We analyze the existence, uniqueness and stability of the evolutionary equilibrium in the evolutionary game.

\subsubsection{Existence and Uniqueness of Equilibrium}

\begin{theorem}
\label{theorem:bounded}
The first-order derivatives of $f_n^{(z)}(\boldsymbol{x}^{(z)}(t))$ with respect to $x_v^{(z)}(t)$ is bounded for all $v \in \mathcal{N}$.
\end{theorem}

\begin{proof}
For simplicity, we omit the $(z)$ and $(t)$ notations in this proof. The first-order derivative of $f_n^{(z)}(\boldsymbol{x}^{(z)}(t))$ with respect to $x_v^{(z)}(t)$, where $v \in \mathcal{N}$ is given by:
\begin{equation}
\label{eqn:partial}
\frac{\partial f_n(\boldsymbol{x})}{\partial x_v} = \delta\left[\frac{\partial x_n}{\partial x_v}(u_n-\bar{u})  +x_n\left(\frac{\partial u_n}{\partial x_v}-\frac{\partial \bar{u}}{\partial x_v}\right) \right].
\end{equation}
The derivative $\frac{\partial u_n}{\partial x_v}$ in Equation~(\ref{eqn:partial}) is computed as follows:
\begin{equation}
\label{eqn:decreasing}
\frac{\partial u_n}{\partial x_v} = \gamma_n \left(\frac{\frac{\partial x_n}{\partial x_v}d_z}{\sum_{z=1}^{Z}d_z x_n^{(z)}} - \frac{x_n d_z^2}{\left[\sum_{z=1}^{Z}d_z x_n^{(z)}\right]^2} \right).
\end{equation}
\end{proof}
This shows that the derivatives of the utilities are bounded, implying that $\left|\frac{\partial u_n}{\partial x_v}\right|$ and thus $\left|\frac{\partial \bar{u}_n}{\partial x_v}\right|$ are bounded $\forall v \in \mathcal{N}$. Therefore, the term $\left|\dot{x}_n\right| = \left|\frac{\partial u_n}{\partial x_v}\right|$ is also bounded.

Next, we prove the uniqueness of the evolutionary equilibrium.
\begin{theorem}
\label{theorem:unique}
For any initial population state $\boldsymbol{x}(0) =\boldsymbol{x}_0 \in \mathbb{X}$, there is a unique population distribution state $\boldsymbol{x}(t)$ that evolves over time according to the replicator dynamics defined in Equation~(\ref{eqn:replicator}).
\end{theorem}
\begin{proof}
From Theorem~\ref{theorem:bounded}, we have established that the partial derivative of $f_n^{(z)}(\boldsymbol{x}^{(z)}(t))$ with respect to $\boldsymbol{x}(t)$ is bounded and continuous $\forall \boldsymbol{x}^{(z)} \in \mathbb{X}^{(z)}, \forall z \in \mathcal{Z} and \forall n \in \mathcal{N}, \forall t$. The maximum absolute value of these partial derivatives $f_n^{(z)}(\boldsymbol{x}^{(z)}(t))$, i.e., $\left|\frac{\partial f_n(\boldsymbol{x})}{\partial x_v}\right|$ is a Lipschitz constant~\cite{engwerda2005lq}. By the Mean Value Theorem, there exists a constant $c$ between $x_1^{(z)}(t)$ and $x_2^{(z)}(t)$ such that $\frac{\partial f_n^{(z)}(c)}{\partial x_v} = \frac{\left|f_n^{(z)}(x_1^{(z)}(t))- f_n^{(z)}(x_2^{(z)}(t))\right|}{x_1^{(z)}(t)-x_2^{(z)}(t)}$. Therefore, we can further derive the following relation:
\begin{multline}
\left|f_n^{(z)}(x_1^{(z)}(t))- f_n^{(z)}(x_2^{(z)}(t))\right| \leq \Phi \left|x_1^{(z)}(t)-x_2^{(z)}(t)\right|,\\
x_1^{(z)},x_2^{(z)} \in \mathbb{X}^{(z)}, z \in \mathcal{Z}, \forall t,
\end{multline}
where $\Phi = \max\left\{\left|\frac{\partial f_n^{(z)}(c)}{\partial x_v}\right|\right\}$. Therefore, $f_n^{(z)}(\boldsymbol{x}^{(z)}(t))$ satisfies the global Lipschitz condition~\cite{engwerda2005lq}. Hence, the replicator dynamics  $f_n^{(z)}(\boldsymbol{x}^{(z)}(t))$ described in Equation~(\ref{eqn:replicator}) with initial condition $\boldsymbol{x}(0) =\boldsymbol{x}_0 \in \mathbb{X}$ has a globally unique solution. 
\end{proof}

From Theorems~\ref{theorem:bounded} and~\ref{theorem:unique}, can conclude that the evolutionary equilibrium, representing a stationary point of the replicator dynamics, not only exists but is also unique. The equilibrium can be derived by solving Equation~(\ref{eqn:replicator}) when it is equal to zero, i.e., $\dot{x}_n^{(z)}(t) = f_n^{(z)}(\boldsymbol{x}^{(z)}(t)) = 0, \forall n \in \mathcal{N}$. All FL workers in a population receive the same utilities and are not better off by deviating from the equilibrium strategy at this evolutionary equilibrium. Therefore, they do not have incentive to associate with any other edge server.

\subsubsection{Stability of Evolutionary Equilibrium}

There are two types of equilibriums in the evolutionary game: boundary equilibrium and interior equilibrium. The boundary equilibrium satisfies the following conditions: (i) $x_{n'}^{(z)}(t) = 0$, and (ii) $x_n^{(z)}(t) = 1$, $\forall n' \neq n \in \mathcal{N}$. In particular, all FL workers in population $z$ associate with a single edge server $n$ and do not choose other edge servers. For the interior equilibrium, the population share is between 0 and 1, i.e., $x_n^{(z)}(t) = [0,1]$, $\forall n \in \mathcal{N}$. The boundary equilibrium is unstable, as even minor disturbances cause the system to shift away from this equilibrium state.

Next, we evaluate the stability of the interior evolutionary equilibrium by using the concept of Lyapinov stability~\cite{sastry1999lyapunov}. According to~\cite{slotine1991applied}, a scalar function $G(t)$ of state $t$, with continuous first-order derivatives, is globally asymptotically stable if the following conditions hold:
\begin{itemize}
\item $G(t) \rightarrow \infty$ when $\|t\| \rightarrow \infty$,
\item $G(t)$ is positive definite,
\item $\dot{G}(t)$ is negative definite.
\end{itemize} 

\begin{theorem}
For any initial population state $\boldsymbol{x}(0) \in \mathbb{X}$ where $0 < x_n^{(z)}(0) < 1$, $\forall n \in \mathcal{N}$, the interior evolutionary equilibrium is globally asymptotically stable and the replicator dynamics in Equation~(\ref{eqn:replicator}) converge to this interior evolutionary equilibrium.
\end{theorem}
\begin{proof}
The Lyapunov function is defined as follows:
\begin{equation}
G(\boldsymbol{x}^{(z)}(t)) = (x_n^{(z)*}(t)-x_n^{(z)}(t))^2.
\end{equation}
This function $G(\boldsymbol{x}^{(z)}(t))$ is always positive definite since $G(\boldsymbol{x}^{(z)}(t)) \geq 0$. As $t$ approaches infinity, $G(\boldsymbol{x}^{(z)}(t))$ also approaches infinity. 
Next, we compute the first order derivative of $G(\boldsymbol{x}^{(z)}(t))$ with respect to $t$:
\begin{multline}
\label{eqn:equilibrium}
\frac{\partial G(\boldsymbol{x}^{(z)}(t))}{\partial t} \\
\shoveleft{= \frac{\partial}{\partial t} (x_n^{(z)*}(t)-x_n^{(z)}(t))^2}\\
\shoveleft{= 2 (x_n^{(z)*}(t)-x_n^{(z)}(t))\frac{\partial(x_n^{(z)*}(t)-x_n^{(z)}(t))}{\partial t}} \\
\shoveleft{= -2 (x_n^{(z)*}(t)-x_n^{(z)}(t))\frac{\partial(x_n^{(z)})}{\partial t}}\\
\shoveleft{= -2\delta x_n^{(z)}(t) (x_n^{(z)*}(t)-x_n^{(z)}(t))(u_n^{(z)}(t)-\bar{u}^{(z)}(t))}.
\end{multline}
From Equation~(\ref{eqn:decreasing}), we observe that $\frac{\partial u_n^{(z)}}{\partial x_n^{(z)}}$ is a decreasing function of $x_n^{(z)}$, i.e., $\frac{\partial u_n^{(z)}}{\partial x_n^{(z)}} < 0$.
At equilibrium point where $u_n^{(z)}(t)=\bar{u}^{(z)}(t)$, $\frac{\partial x_n^{(z)}(t) }{\partial t} = 0$. When $u_n^{(z)}(t)<\bar{u}^{(z)}(t)$, $x_n^{(z)*}<x_n^{(z)}$, $\frac{\partial x_n^{(z)}(t) }{\partial t} < 0$. When $u_n^{(z)}(t)>\bar{u}^{(z)}(t)$, $x_n^{(z)*}>x_n^{(z)}$, $\frac{\partial x_n^{(z)}(t) }{\partial t} < 0$. 
Therefore, the first order derivative of $G(\boldsymbol{x}^{(z)}(t))$ with respect to $t$ is always less than or equal to zero, i.e.,
\begin{multline}
\frac{\partial G(\boldsymbol{x}^{(z)}(t))}{\partial t} \\
=  -2\delta x_n^{(z)}(t) (x_n^{(z)*}(t)-x_n^{(z)}(t))(u_n^{(z)}(t)-\bar{u}^{(z)}(t)) \leq 0.
\end{multline}
Since $\dot{G}(\boldsymbol{x}^{(z)}(t))$ is negative definite, this satisfies the stability conditions for Lyapunov functions. Hence, the replicator dynamics in Equation~(\ref{eqn:replicator}) converge to the globally asymptotically stable interior evolutionary equilibrium.
\end{proof}

\subsection{Algorithm of the Evolutionary Game}

Given that there are $Z$ populations of FL workers, an initial population state $\boldsymbol{x}^{(z)}(0)$ is first initialized where the FL workers are randomly assigned to one of the $N$ edge servers (line 2). The evolutionary edge association strategies of these FL workers are governed by the replicator dynamics defined in Equation~(\ref{eqn:replicator}). Each FL worker in population $z \in \mathcal{Z}$ adjusts its edge association strategy to maximize its utility. Specifically, given the current population state $\boldsymbol{x}^{(z)}(t)$, an FL worker may leave its current edge server $n$ and join another edge server $n'$ (where $n \neq n'$) if it achieves higher utility by joining edge server $n'$. The edge server association strategy is illustrated from the perspective of an FL worker in population $z$. In order to compute the replicator dynamics defined in Equation~(\ref{eqn:replicator}), we first calculate the utility of FL workers of population $z$ for associating with each edge server $n$ ($\forall n \in \mathcal{N}$)~(line 8). Based on this, the average utility of FL workers in population $z$ can be computed~(line 10). Then, the replicator dynamics and the population shares of population $z$ can also be computed~(line 12). If the $\dot{x}_n^{(z)}(t) > 0$, it indicates that the FL workers that are associated with edge server $n$ achieve a higher utility than the average utility, hence attracting more FL workers to join this particular edge server $n$~(lines 13-14). Conversely, if $\dot{x}_n^{(z)}(t) < 0$, it implies that the FL workers that are associated with edge server $n$ achieve a lower utility than the average utility, prompting them to change their edge association strategies and switch to other edge servers $n' \neq n$~(lines 15-16). The algorithm continues until the evolutionary equilibrium is reached~(line 22) where the FL workers have no incentive to change their edge association strategy.

\begin{algorithm}[t]
\caption{Algorithm for Evolutionary Game for Edge Servers' Cluster Formation.}
\footnotesize
\label{algo}
\begin{algorithmic}[1]
 \renewcommand{\algorithmicrequire}{\textbf{Input:}}
 \renewcommand{\algorithmicensure}{\textbf{Output:}}
 \REQUIRE Set of populations, $\mathcal{Z}=\{1, \ldots,z,\ldots,Z\}$, set of edge servers, $\mathcal{N}=\{1,\ldots, n,\ldots, N\}$ 
 \ENSURE Final population state $\boldsymbol{x}^{(z)*}(t) = [x^{(z)*}_1(t), \ldots, x^{(z)*}_n(t), \ldots, x^{(z)*}_N(t)]$
 
\STATE t=0
\STATE Initialize a initial population state $\boldsymbol{x}^{(z)}(0) = [x^{(z)}_1(0), \ldots, x^{(z)}_n(0), \ldots, x^{(z)}_N(0)]$

\STATE \textbf{\emph{\underline{Replicator Dynamics:}}}
 
\WHILE {$\boldsymbol{x}^{(z)*}(t) \neq \boldsymbol{x}^{(z)}(t)$}
 	\STATE Update the final population state such that $\boldsymbol{x}^{(z)*}(t)=\boldsymbol{x}^{(z)}(t-1)$
	\FOR {\textbf{each} population $z \in \mathcal{Z}$ (population share of population $z$ is associated with edge server $n$, i.e., $x ^{(z)}_n(t) \in$ population state $\boldsymbol{x}^{(z)}(t)$}
		\FOR  {\textbf{each} edge server $n \in \mathcal{N}$}
			\STATE Compute $u_n^{(z)}(t)$ 
		\ENDFOR
		\STATE Compute $\bar{u}^{(z)}(t)$
		\FOR  {\textbf{each} edge server $n \in \mathcal{N}$}
			\STATE Compute ${x}_n^{(z)}(t)$ and $\dot{x}_n^{(z)}(t)$
			\IF {$\dot{x}_n^{(z)}(t) > 0$ }
				\STATE Other FL workers in population $z$ are attracted to choose this particular edge server $n$
				\ELSIF {$\dot{x}_n^{(z)}(t) < 0$}
					\STATE FL workers in population $z$ chooses other edge servers $n' \neq n$
			\ENDIF
			\STATE Update the current population state $\boldsymbol{x}^{(z)}(t)$
		\ENDFOR
	\ENDFOR
\ENDWHILE

\RETURN Final population state $\boldsymbol{x}^{(z)*}(t) = [x^{(z)*}_1(t), \ldots, x^{(z)*}_n(t), \ldots, x^{(z)*}_N(t)]$ that is at evolutionary equilibrium and stable
\end{algorithmic}
\end{algorithm}

\section{Simulation Results}
\label{sec:evaluation}

In this section, we discuss the evolutionary game in the cluster formation of edge servers and present the evaluation of synthetic-data-empowered HFL framework. Table~\ref{tab:simulation} summarizes the values of the simulation parameters.

The HFL network consists of an FL server, 3 edge servers and 50 FL workers. The FL server trains an FL model for an image classification task. The FL server selects edge servers based on client selection algorithms~\cite{nishio2019client, abdulrahman2021fedmccs} or resource allocation schemes~\cite{lim2022decentralized,lim2021dynamic}. The edge servers, in turn,  employ the FL workers to facilitate the FL training task. The FL workers are rewarded based on their data contributions to the FL task.

\begin{table}[h]
\caption{Simulation Parameter Values.} 
\label{tab:simulation}
\centering
\renewcommand\arraystretch{1.5}
\begin{tabularx}{8.7cm}{XlX}
\hline
\hline
\textbf{Parameter}& \textbf{Values}\\ [0.5ex]
\hline
Number of FL workers, $J$ & $50$\\
Number of training iterations, $K$ & $1000$\\
Number of edge servers, $N$ & $3$\\
Amount of computation resources, $c_z$ & [10, 50]\\
Data quantity of FL worker in population $z$, $d_z$ & [2000, 4000]\\
Amount of communication resources, $m_z$ & [10, 50]\\
Amount of computation resources for edge server $n$, $s_n$ & [2, 6]\\
Unit computation cost, $\alpha$ & 0.001\\
Unit communication cost, $\beta$ & 0.001\\
Reward pool offered by edge server $n$, $\gamma_n$ & [100, 900]\\
Learning rate, $\delta$ & [0.001, 0.1]\\
\hline
\end{tabularx}
\end{table}


\subsection{Evolutionary Game}

\begin{figure}
\centering
\includegraphics[width=\linewidth]{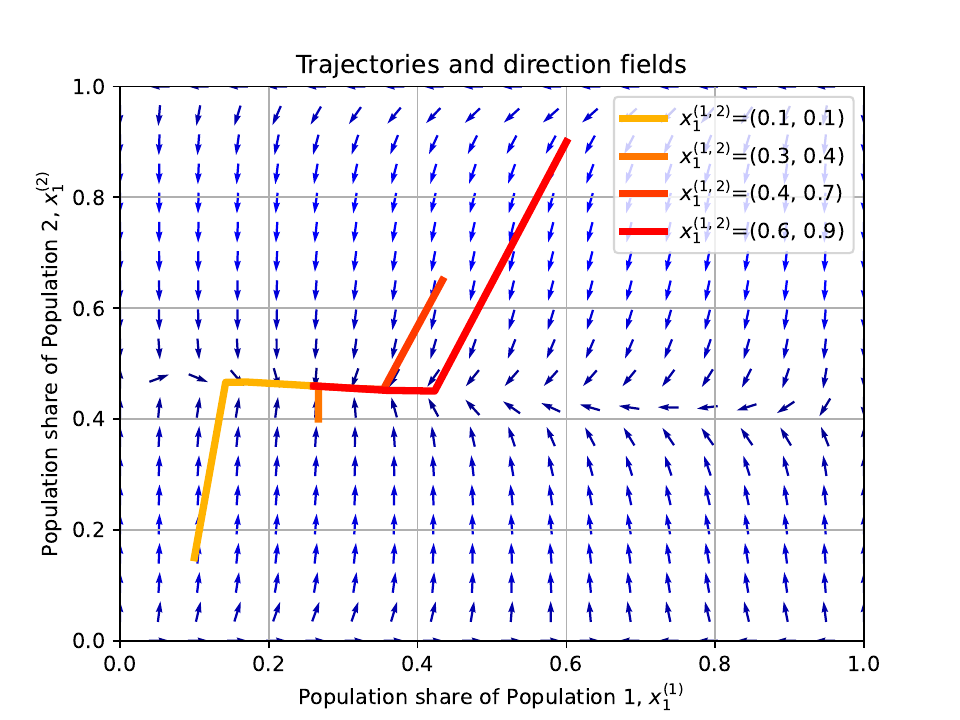}
\caption{Phase plane of the replicator dynamics.}
\label{fig:phaseplane}
\end{figure}

\begin{figure*}
     \centering
     \begin{multicols}{3}
     \begin{subfigure}[b]{\linewidth}
         \centering
         \includegraphics[width=\linewidth]{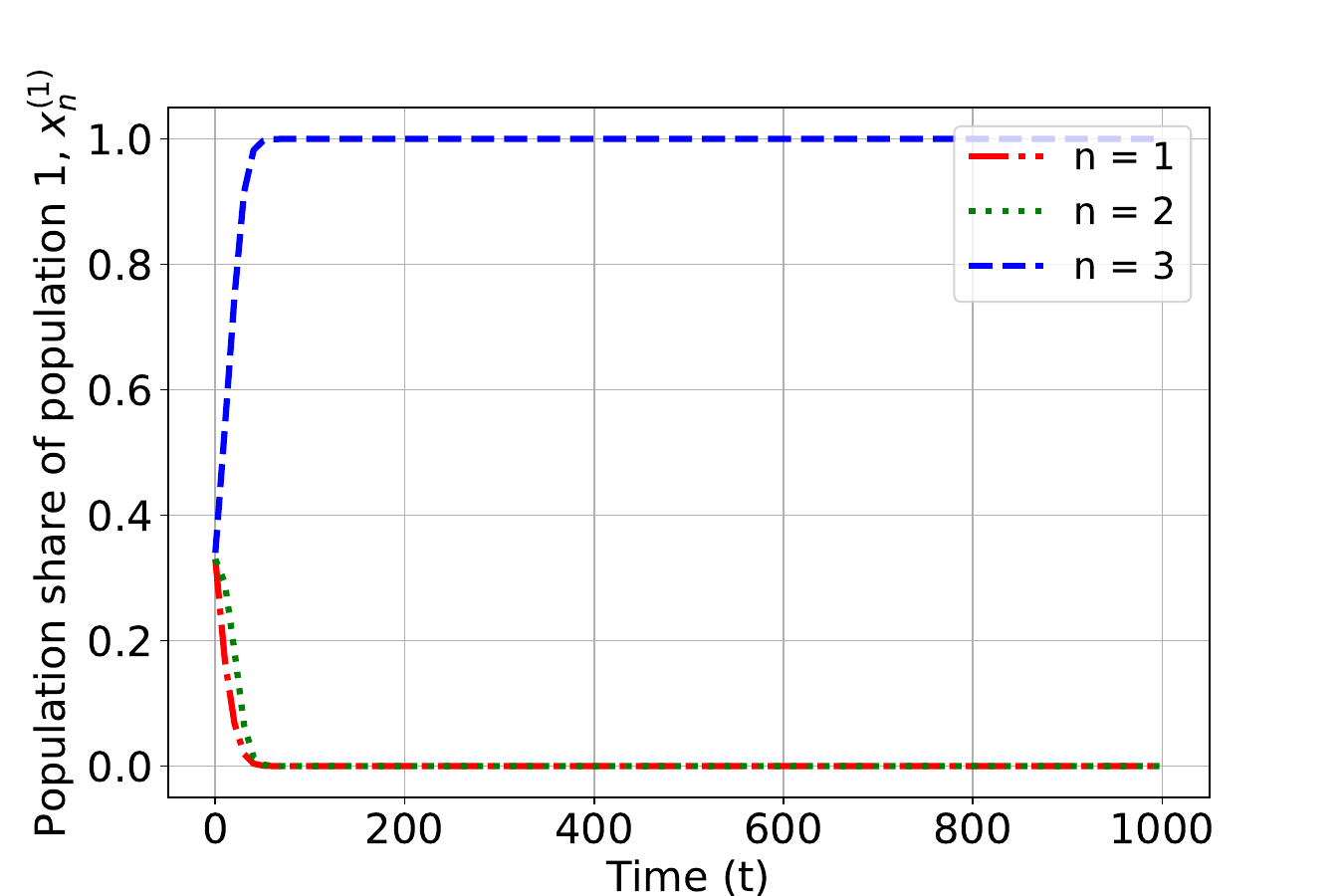}
         \caption{Population shares of population 1 for different edge servers.}
         \label{fig:population1}
     \end{subfigure}
     \begin{subfigure}[b]{\linewidth}
         \centering
         \includegraphics[width=\linewidth]{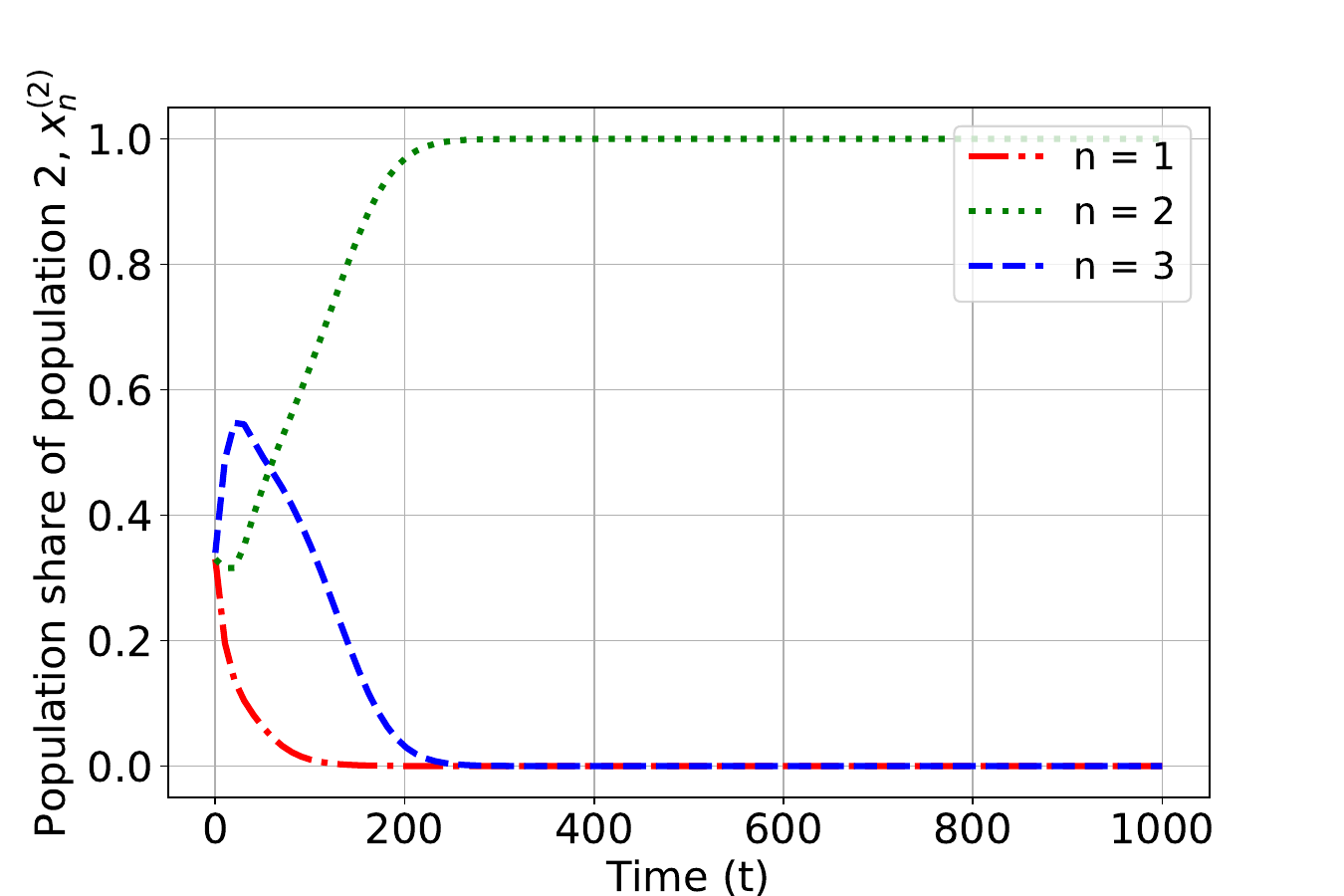}
         \caption{Population shares of population 2 for different edge servers.}
         \label{fig:population2}
     \end{subfigure}
     \begin{subfigure}[b]{\linewidth}
        \centering
         \includegraphics[width=\linewidth]{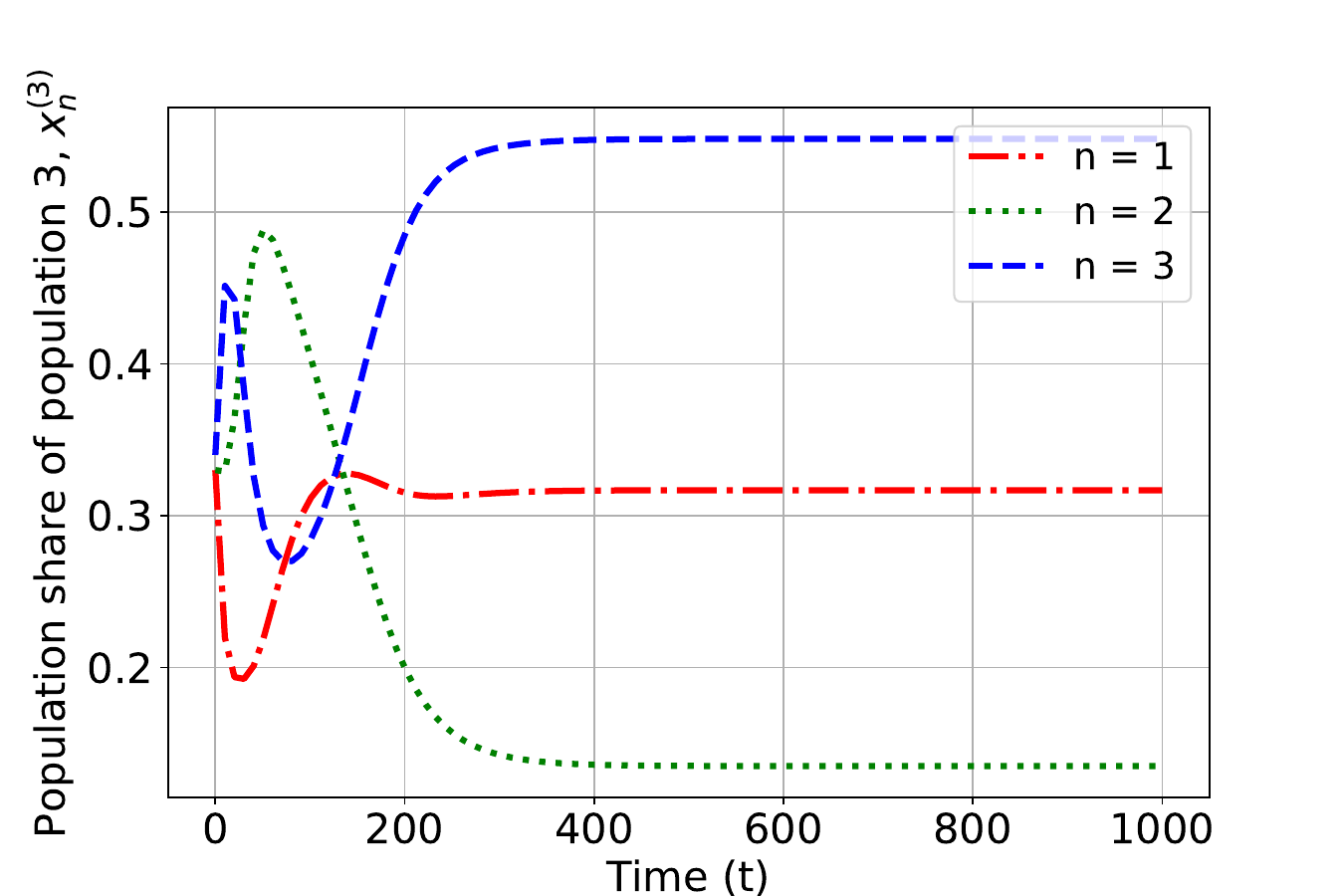}
         \caption{Population shares of population 3 for different edge servers.}
         \label{fig:population3}
     \end{subfigure}
\end{multicols}
        \caption{Population shares of different populations.}
        \label{fig:populations}
\end{figure*}

\begin{figure*}
\centering
\begin{multicols}{3}
\includegraphics[width=\columnwidth]{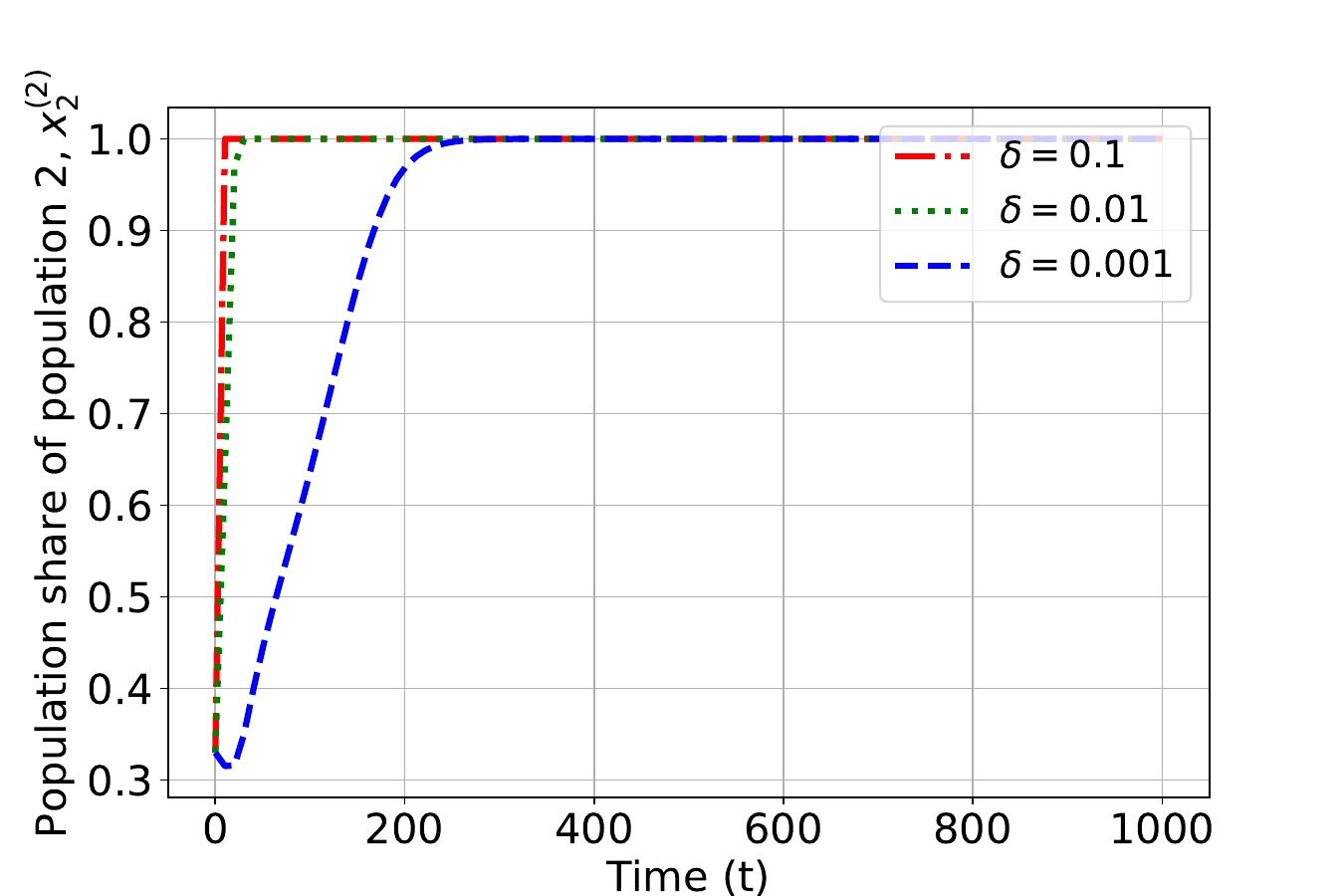}
	\caption{Population share of population $2$ given different learning rates.}
	\label{fig:learningrate}
\includegraphics[width=\columnwidth]{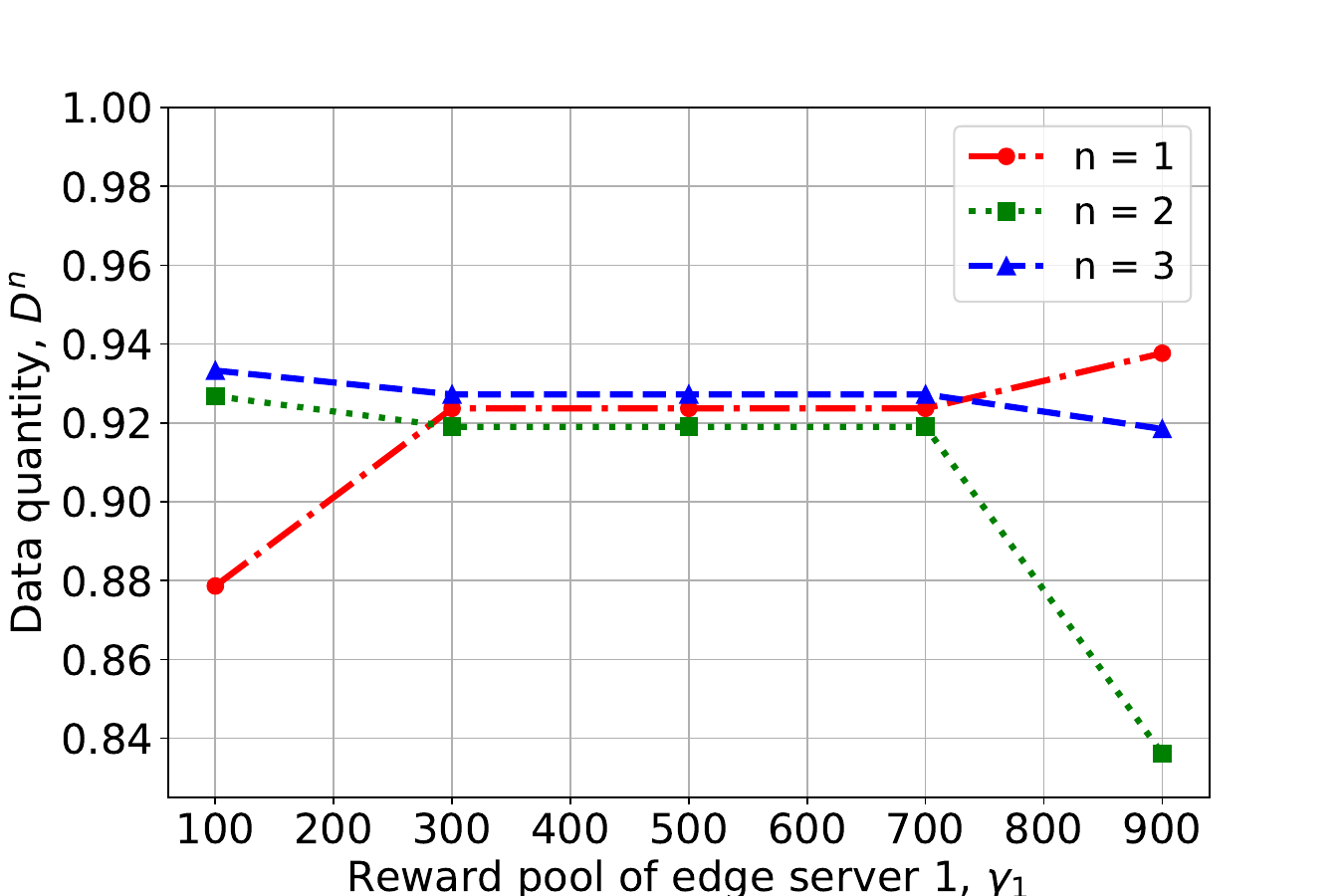} 
    	\caption{Data quantities given different reward pools.}
    	\label{fig:rewardpool} 
\includegraphics[width=\columnwidth]{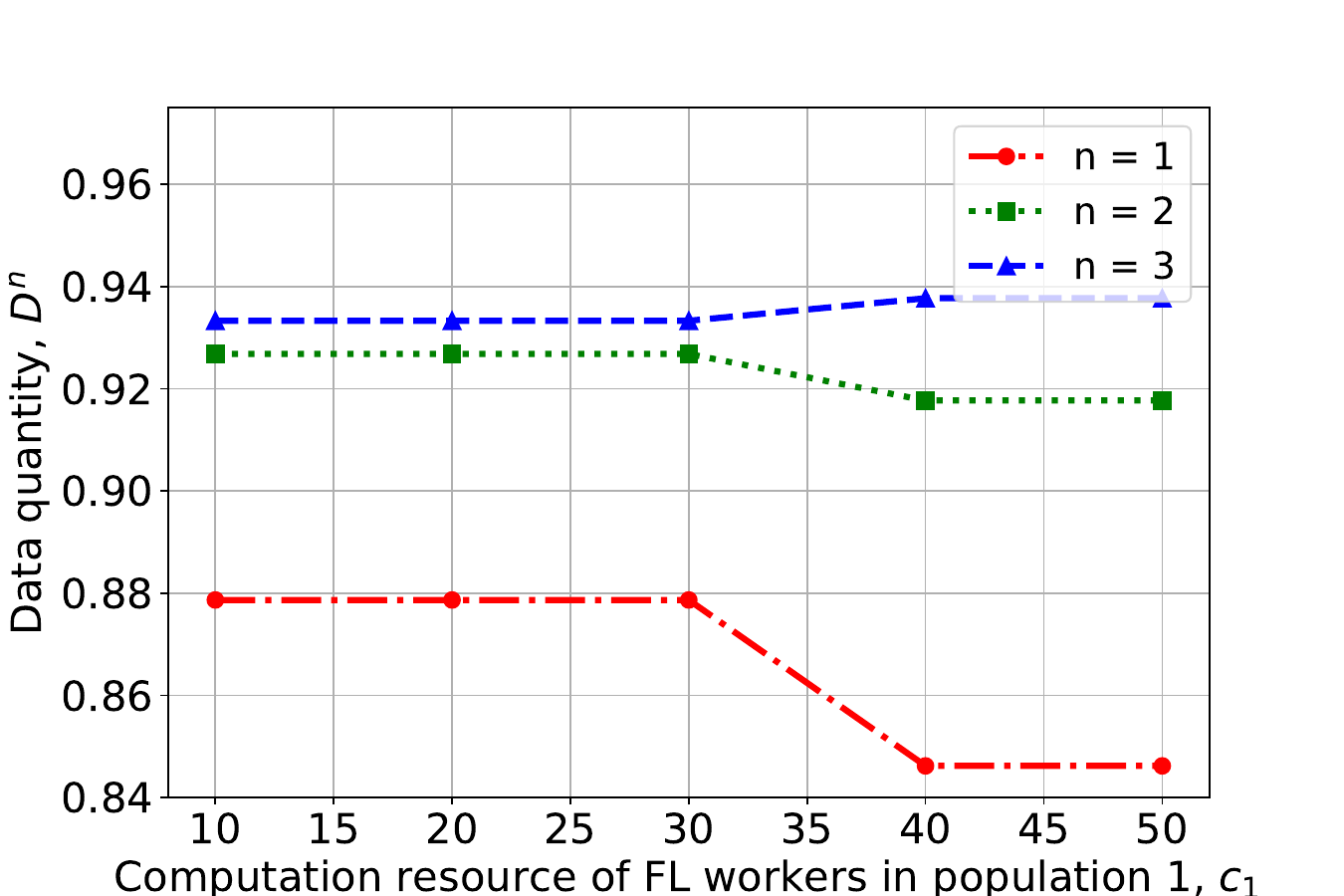} 
    	\caption{Data quantities given different computation costs.}
    	\label{fig:computationcost} 
\end{multicols}
\end{figure*}

We first analyze the phase plane of the replicator dynamics to demonstrate the uniqueness of the evolutionary equilibrium. For simplicity, we consider two populations of FL workers where FL workers in populations $1$ and $2$ have $2,000$ and $4,000$ data samples respectively, i.e., $d_1=2000$ and $d_2=4000$. The FL workers need to decide between two strategies, i.e., joining either edge server $1$ or $2$ with computation requirements of $s_1= 2$ or $s_2=4$, respectively. Figure~\ref{fig:phaseplane} shows the evolution of the population states of the two populations in joining edge server $1$ given varying initial conditions. In particular, $x_1^{(1,2)} = (0.1, 0.1)$ means $10\%$ of FL workers in population $1$ and $10\%$ of FL workers in population $2$ are associated with edge server $1$. From Fig.~\ref{fig:phaseplane}, we observe that despite the different initial conditions, e.g., $x_1^{(1,2)} = (0.6, 0.9)$ or $x_1^{(1,2)} = (0.1, 0.1)$, the evolutionary equilibrium always converges to a unique solution.  

Next, we consider a more complex HFL network with three populations of FL workers and three edge servers. The edge servers have an increasing computation requirement, i.e., $s_1=2 < s_2=4 <s_3=6$. To offset the greater computational demands placed on the FL workers, the edge servers provide progressively higher rewards. In particular, $\gamma_1=100<\gamma_2=300<\gamma_3=500$. The FL workers have different resource costs, with population $1$ being the most least expensive to employ and population $3$ being the most expensive, i.e., $c_3 = 50 > c_2 = 30 > c_1 = 10$ and $m_3 = 50 > m_2 = 30 > m_1 = 10$. The data quantities of the FL workers are fixed at $d_1=d_2=d_3=3,000$. We then study the stability of the equilibrium solution. From Figs.~\ref{fig:population1}-\ref{fig:population3}, we observe that the population shares of FL workers in all populations eventually converge with time. This implies that once the evolutionary equilibrium is reached, all FL workers do not have incentive to switch from one edge server to another since they are unable to increase their utilities by changing their edge association strategies. This confirms the stability of the evolutionary equilibrium. From Figs.~\ref{fig:population1}-~\ref{fig:population3} that show the evolution of population shares of different populations, we observe the utility-maximizing behaviour of the FL workers. Specifically, not all FL workers opt for edge server $3$, even though it offers the largest reward. For example, all FL workers in population $2$ select edge server $2$, as the higher computational demands of edge server $3$ make it less attractive despite its greater reward pool. Hence, the FL workers in population $2$ are able to achieve higher utility by choosing edge server $2$ instead of edge server $3$. Similarly, the FL workers in population $3$ are distributed among three edge servers with 31\%, 14\% and 55\% of the FL workers in population $3$ choosing edge server $1$, $2$ and $3$ respectively. This is because the reward pools of the edge servers are shared among all FL workers in their clusters. As such, when too many FL workers select a particular edge server, a crowding effect occurs. Some FL workers may achieve higher utility by choosing other edge servers.

\begin{figure*}
\centering
\begin{multicols}{4}
\begin{subfigure}[b]{\linewidth}
         \centering
	\includegraphics[width=\linewidth]{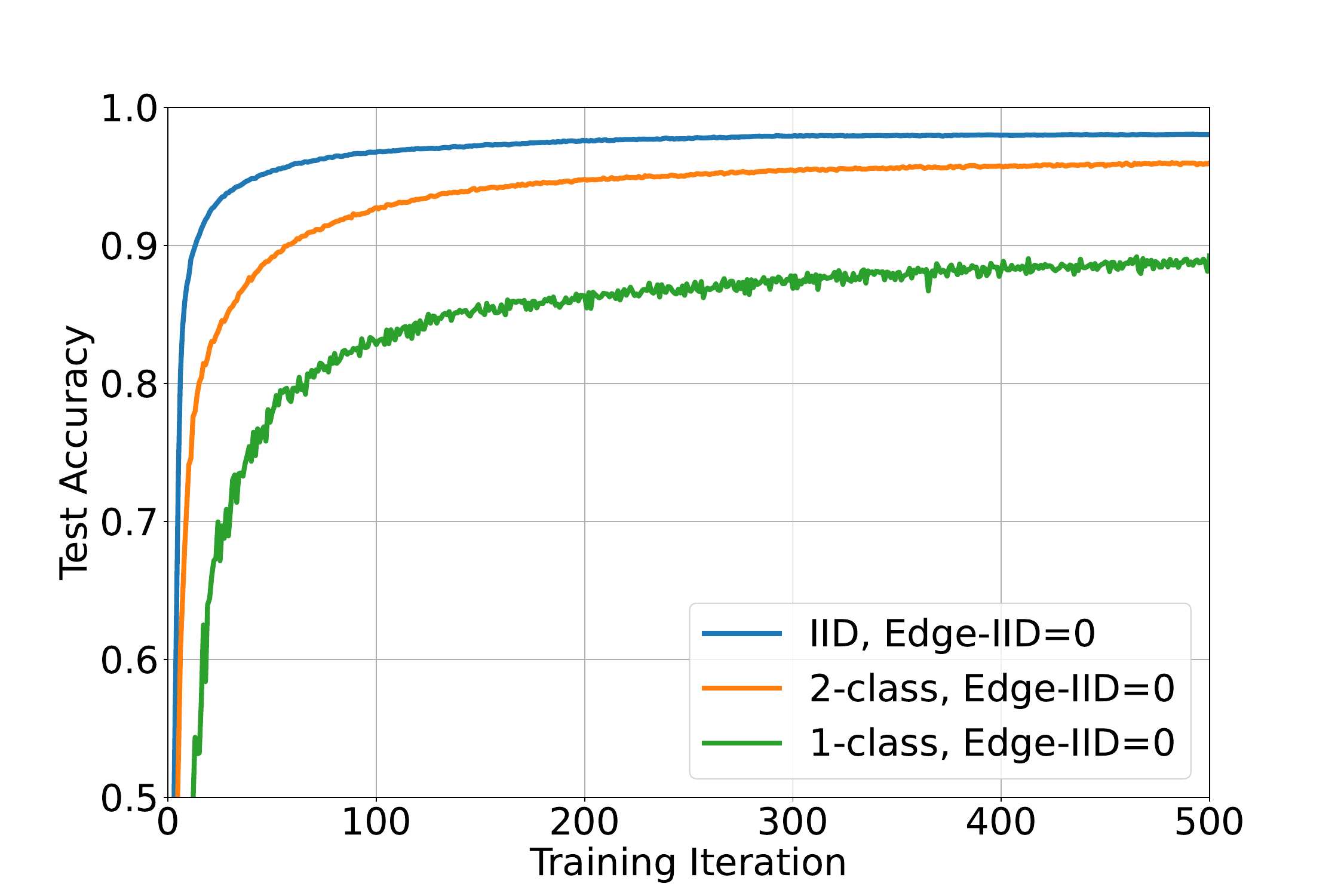}
	\caption{FL model accuracy on \emph{MNIST} dataset when edge servers are non-IID.}
	\label{fig:reduceaccuracy_edgeiid0}
	\end{subfigure}
\begin{subfigure}[b]{\linewidth}
         \centering
	\includegraphics[width=\linewidth]{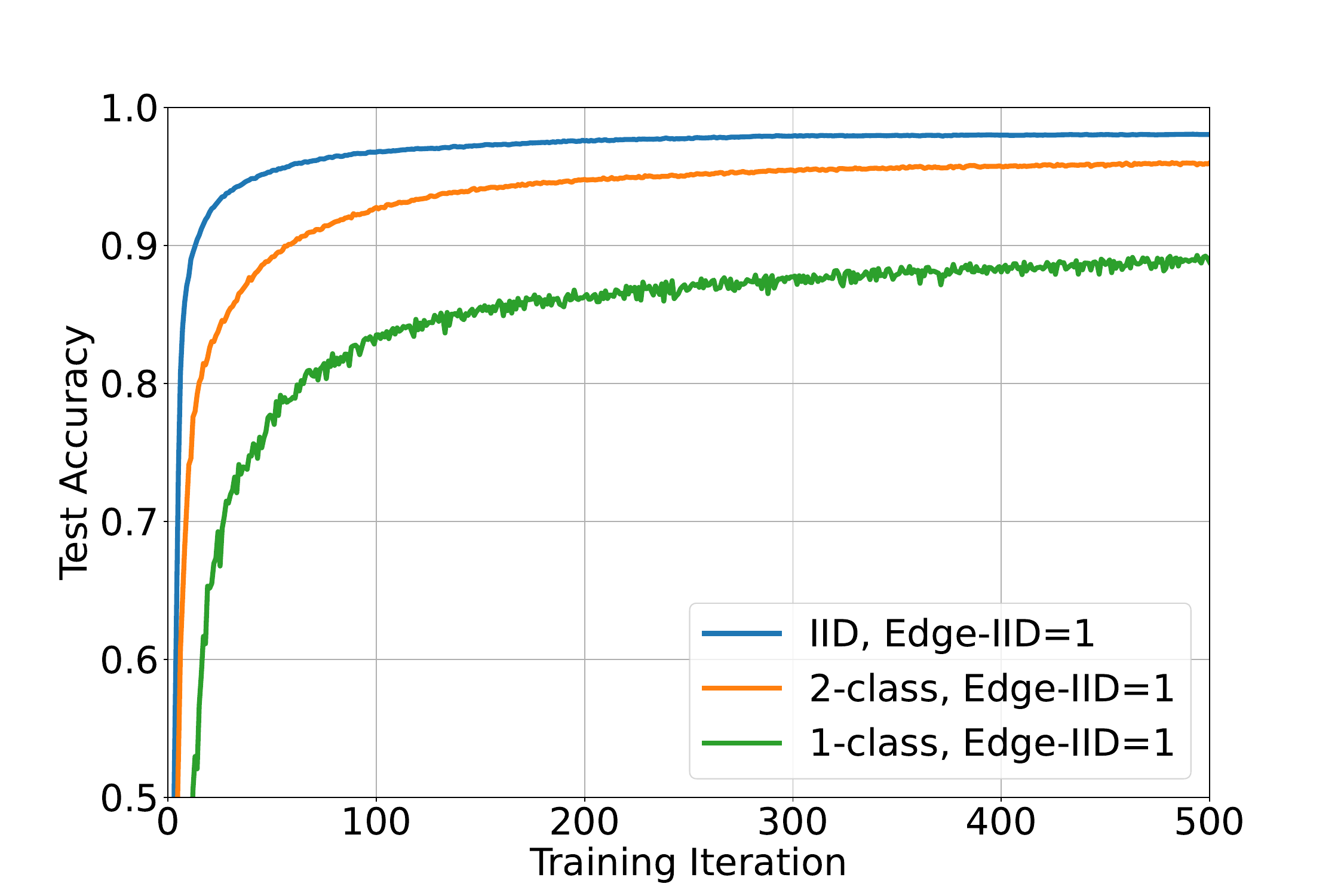}
	\caption{FL model accuracy on \emph{MNIST} dataset when edge servers are IID.}
	\label{fig:reduceaccuracy_edgeiid1}
	\end{subfigure}
\begin{subfigure}[b]{\linewidth}
         \centering
	\includegraphics[width=\linewidth]{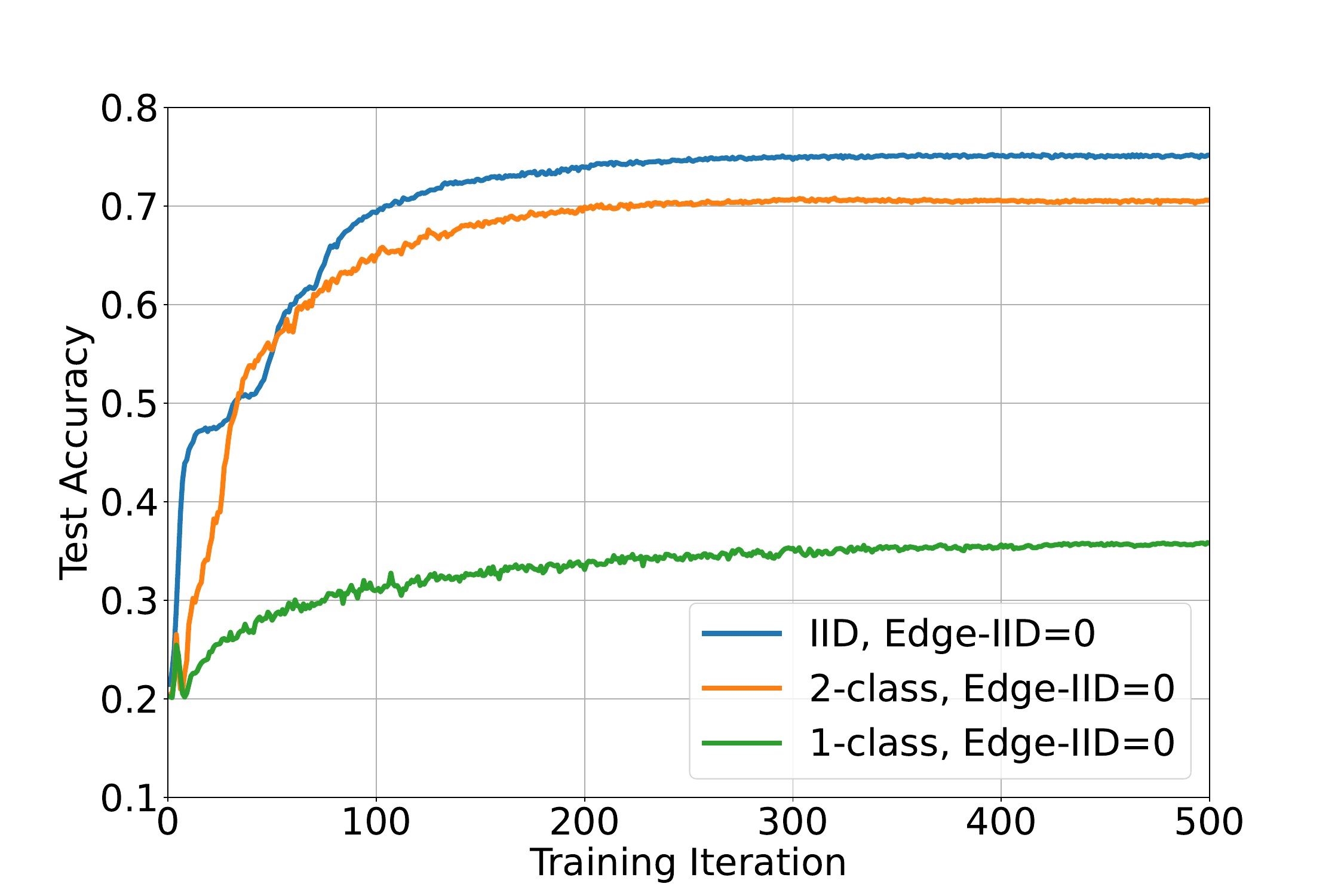}
	\caption{FL model accuracy on \emph{CIFAR-10} dataset when edge servers are non-IID.}
	\label{fig:cifar10-reduceaccuracy-edgeiid0}
	\end{subfigure}
\begin{subfigure}[b]{\linewidth}
         \centering
	\includegraphics[width=\linewidth]{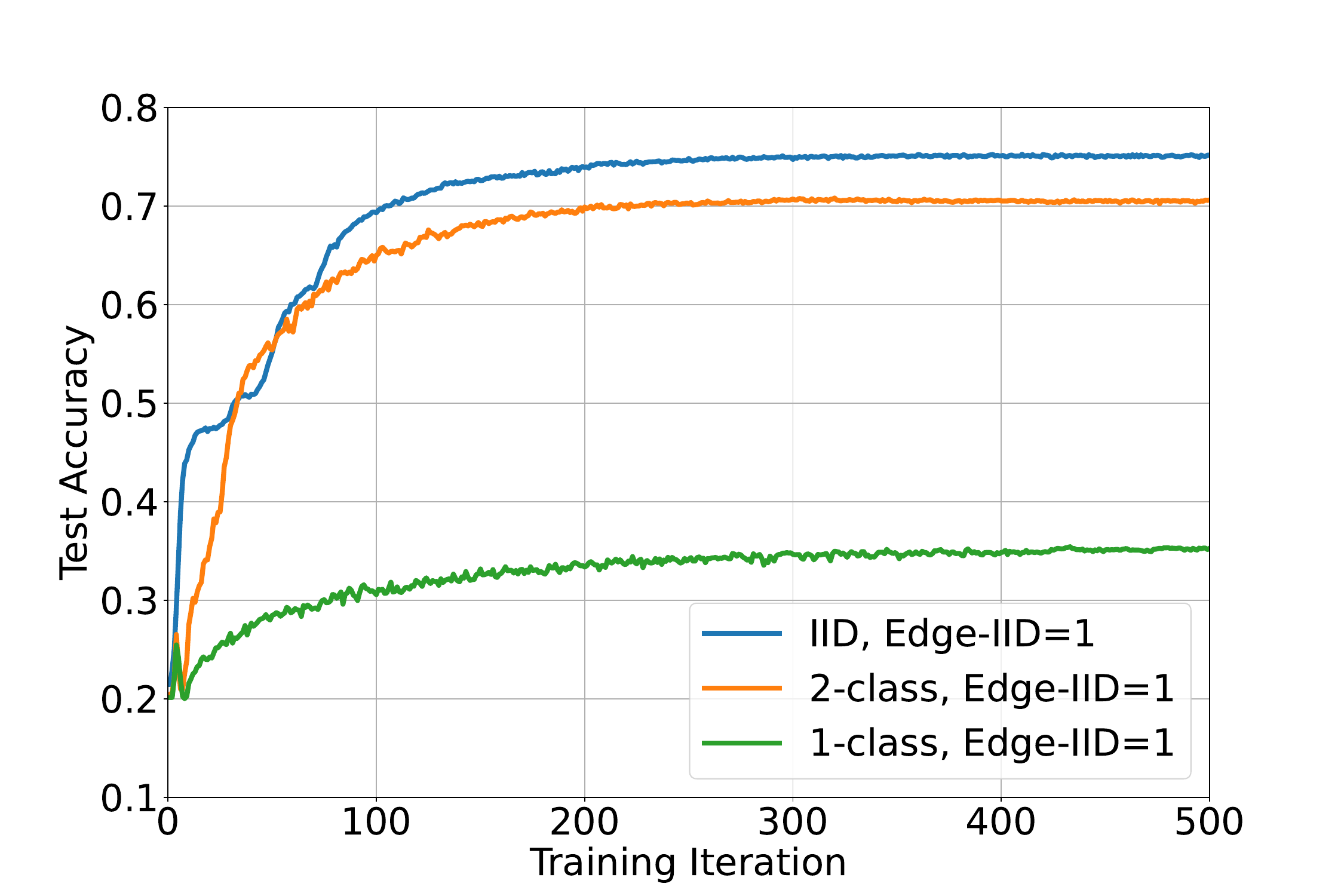}
	\caption{FL model accuracy on \emph{CIFAR-10} dataset when edge servers are IID.}
	\label{fig:cifar10-reduceaccuracy-edgeiid1}
	\end{subfigure}
\end{multicols}
	\caption{FL model accuracy given various non-IID scenarios.}
        \label{fig:nonIIDaccuracy}
\end{figure*}

\begin{figure*}
\centering
\begin{multicols}{3}
\begin{subfigure}[b]{\linewidth}
         \centering
	\includegraphics[width=\linewidth]{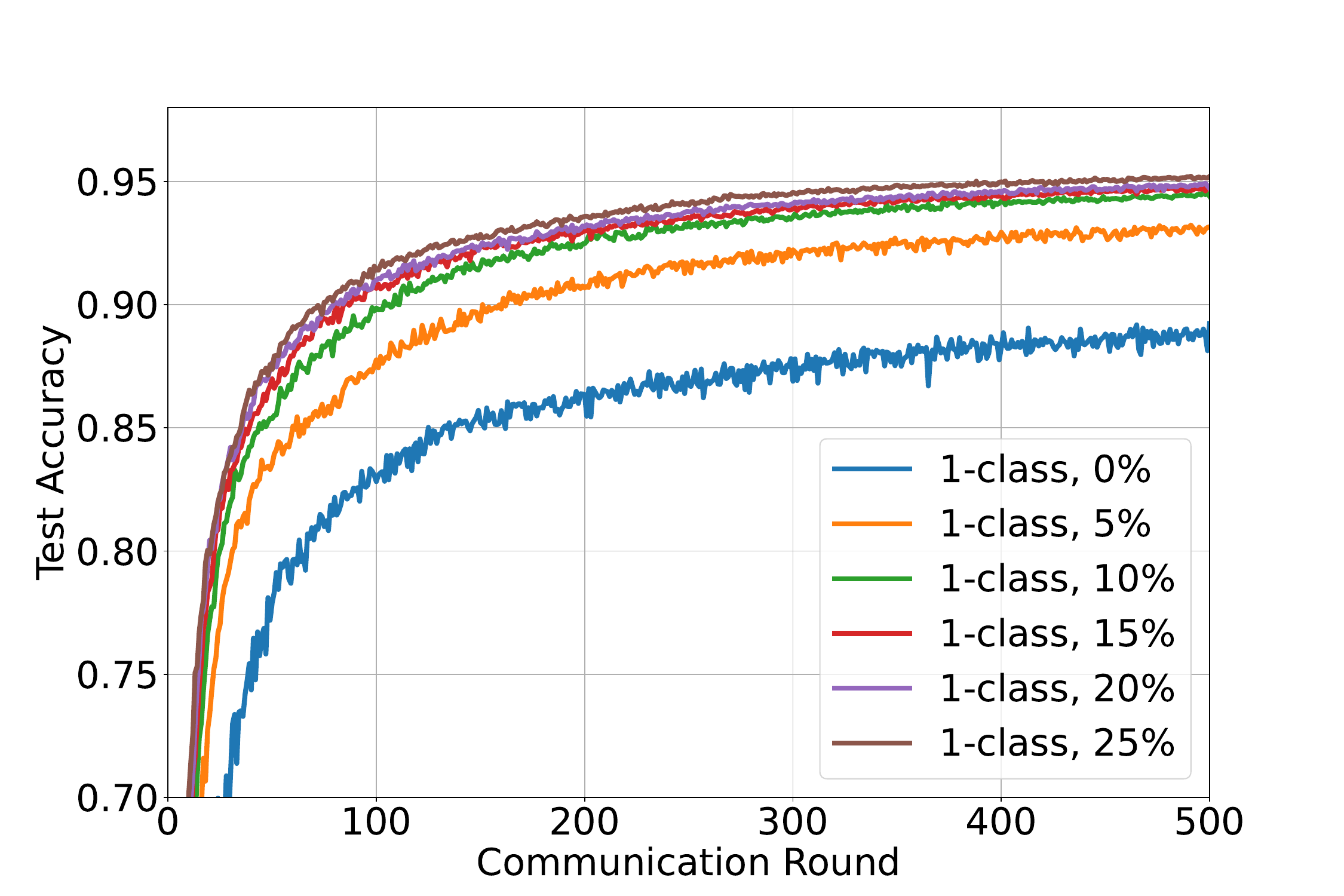}
	\caption{FL model accuracy on \emph{MNIST} dataset with varying amount of synthetic data samples when the edge servers are non-IID, each FL worker has a single class of data samples.}
	\label{fig:differentpercent1classedgeiid0}
	\end{subfigure}
\begin{subfigure}[b]{\linewidth}
         \centering
	\includegraphics[width=\linewidth]{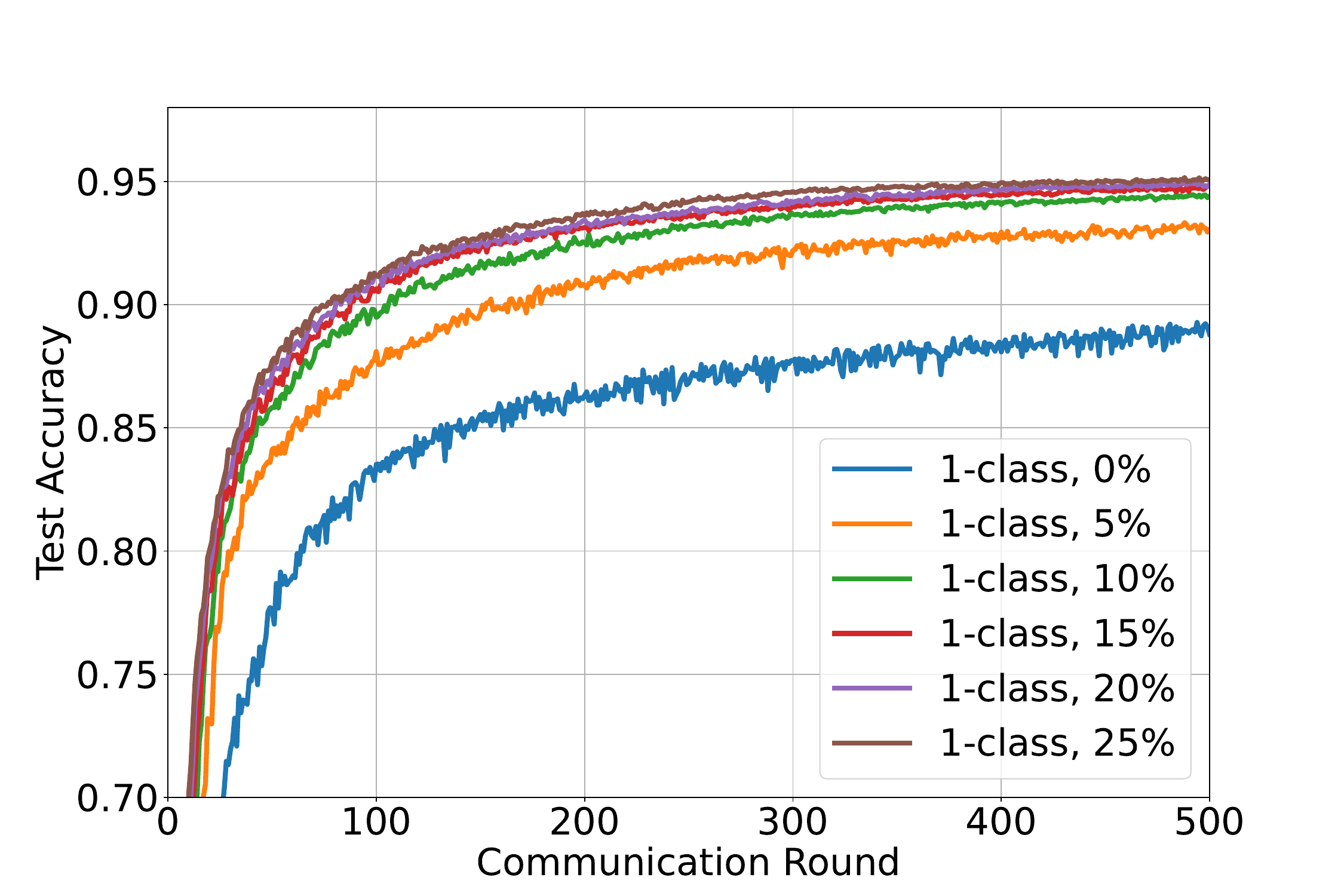}
	\caption{FL model accuracy on \emph{MNIST} dataset with varying amount of synthetic data samples when the edge servers are IID, each FL worker has a single class of data samples.}
	\label{fig:differentpercent1classedgeiid1}
	\end{subfigure}
\begin{subfigure}[b]{\linewidth}
         \centering
	\includegraphics[width=\linewidth]{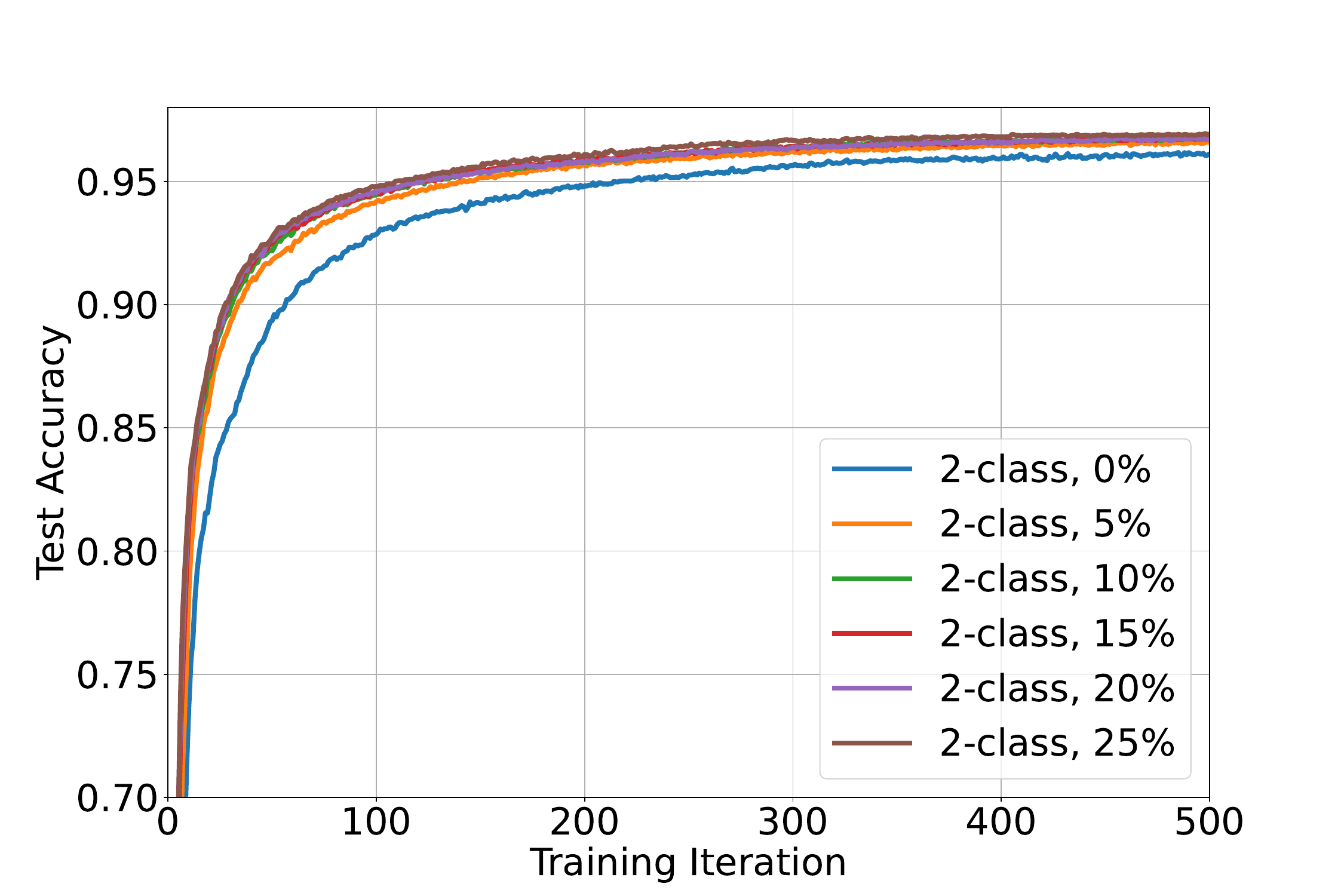}
	\caption{FL model accuracy on \emph{MNIST} dataset with varying amount of synthetic data samples when the edge servers are IID, each FL worker has two classes of data samples.}
	\label{fig:differentpercent2class}
	\end{subfigure}
\end{multicols}
	\caption{FL model accuracy on \emph{MNIST} dataset given varying amount of synthetic data samples under various non-IID scenarios.}
        \label{fig:diffferentpercentmnist}
\end{figure*}

We also study the effect of different learning rates on the evolutionary equilibrium. As shown in Fig.~\ref{fig:learningrate}, despite varying learning rates, the proportion of FL workers in population $2$ selecting edge server $2$ consistently stabilizes at the same point, highlighting the uniqueness and stability of the evolutionary equilibrium. Furthermore, we observe that the learning rate primarily affects the speed of convergence which depends on how fast the FL workers adapt their strategies, e.g., some FL workers may gain accurate information about the utilities of other FL workers. A higher learning rate results in a faster convergence rate.

Besides, we explore the effect of different reward pools and computation costs on the aggregated data quantities of edge servers. In Fig.~\ref{fig:rewardpool}, we observe that as the amount of reward pool offered by edge server $1$ increases, the aggregated data quantity of edge server $1$ increases while the aggregated data quantities of edge servers $2$ and $3$ decrease. This is because the FL workers that are associated with edge servers $2$ and $3$ may change their strategies and choose edge server $1$ instead. Figure~\ref{fig:computationcost} shows that as the computation resource required by FL workers in population $1$ to train on local datasets increases, the aggregated data quantity of edge server $1$ decreases. However, the aggregated data quantities of edge servers $2$ and $3$ decreases and increases respectively. The increase in computation resource required by FL workers in population $1$ causes a decrease in utility when the FL workers are associated with edge server $1$, prompting them to choose other edge servers. The shift in edge association strategies of these FL workers affects the utilities of other FL workers, potentially causing them to change their strategies as well. This illustrates the interdependent interactions among FL workers in the HFL network.

\subsection{Synthetic-data-empowered Hierarchical Federated Learning}

For the FL training, we use the standard 10-class handwritten digit classification dataset \emph{MNIST} and 10-class object classification dataset \emph{CIFAR-10}. The \emph{MNIST} dataset consists of $70,000$ data samples that is split into $60,000$ training samples and $10,000$ test samples. The FL model is built using a Convolutional Neural Network (CNN) with $21,840$ trainable parameters, following the architecture in~\cite{mcmahan2017communication}. For local training, each FL worker applies mini-batch Stochastic Gradient Descent (SGD) with a batch size of 20 and an initial learning rate of $0.01$, which decays exponentially at a rate of $0.995$ with each iteration. The \emph{CIFAR-10} dataset consists of $60,000$ data samples that is split into $50,000$ training samples and $10,000$ test samples. For the \emph{CIFAR-10} dataset, the FL model is trained using the CNN with $5,852,170$ parameters. Mini-batch SGD with a batch size of $20$, an initial learning rate of $0.1$ and an exponential learning rate decay of $0.992$ per training iteration is employed for the local computation of FL training of each FL worker. 


When the FL workers no longer able to increase their utilities by altering their edge server association strategies, the edge servers have successfully formed their clusters of FL workers to facilitate the FL training task. However, the non-IID nature of the data continues to pose a challenge, potentially degrading the accuracy of the FL model. This challenge can be mitigated by providing FL workers with synthetic datasets relevant to the FL training task. 


We examine two types of non-IID data distributions among FL workers. In the first type, each FL worker possesses data samples from only a single class, while in the second type, each FL worker has data samples from any two of the ten classes. The FL workers are then randomly assigned to edge servers, resulting in scenarios where the data distributions at the edge servers are either IID or non-IID.

For both cases, i.e., whether the edge servers maintain IID or non-IID distributions, we observe that the accuracy of the FL model declines as the non-IID nature of the FL workers’ data increases. This trend is evident in the \emph{MNIST} dataset, as illustrated in Figs~\ref{fig:reduceaccuracy_edgeiid0} and \ref{fig:reduceaccuracy_edgeiid1}. Similarly, for the \emph{CIFAR-10} dataset, the accuracy degradation follows the same pattern, as shown in Figs~\ref{fig:cifar10-reduceaccuracy-edgeiid0} and \ref{fig:cifar10-reduceaccuracy-edgeiid1}. These findings highlight the effect of non-IID data distributions on FL model performance, emphasizing the need for techniques such as synthetic data augmentation to mitigate accuracy loss. For simulations that involve \emph{MNIST} dataset, the synthetic dataset are generated using cGAN proposed in~\cite{rahmdel2024cgan}. Similar to the \emph{MNIST} dataset, this synthetic dataset has $70,000$ data samples where $60,000$ and $10,000$ are training and testing samples, respectively. For simulations that involve \emph{CIFAR-10} dataset, we use \emph{CIFAKE} dataset~\cite{bird2024cifake} that is generated by using LDMs. Similar to the \emph{CIFAR-10} dataset, this \emph{CIFAKE} dataset has $60,000$ data samples where $50,000$ and $10,000$ are training and testing samples, respectively.

Figures~\ref{fig:differentpercent1classedgeiid0}-\ref{fig:differentpercent2class} illustrate the effect of incorporating synthetic datasets on FL model accuracy for the \emph{MNIST} dataset across different non-IID scenarios.

\begin{itemize}
\item \emph{Scenario 1: }Each FL worker has data samples from two different classes, and the edge servers maintain an IID data distribution.
\item \emph{Scenario 2: }Each FL worker has data samples belonging to just one class, while the edge servers still have an IID distribution.
\item \emph{Scenario 3: }Each FL worker has data samples belonging to just one class, but the edge servers have a non-IID data distribution.
\end{itemize} 

\begin{figure}
\centering
\includegraphics[width=\linewidth]{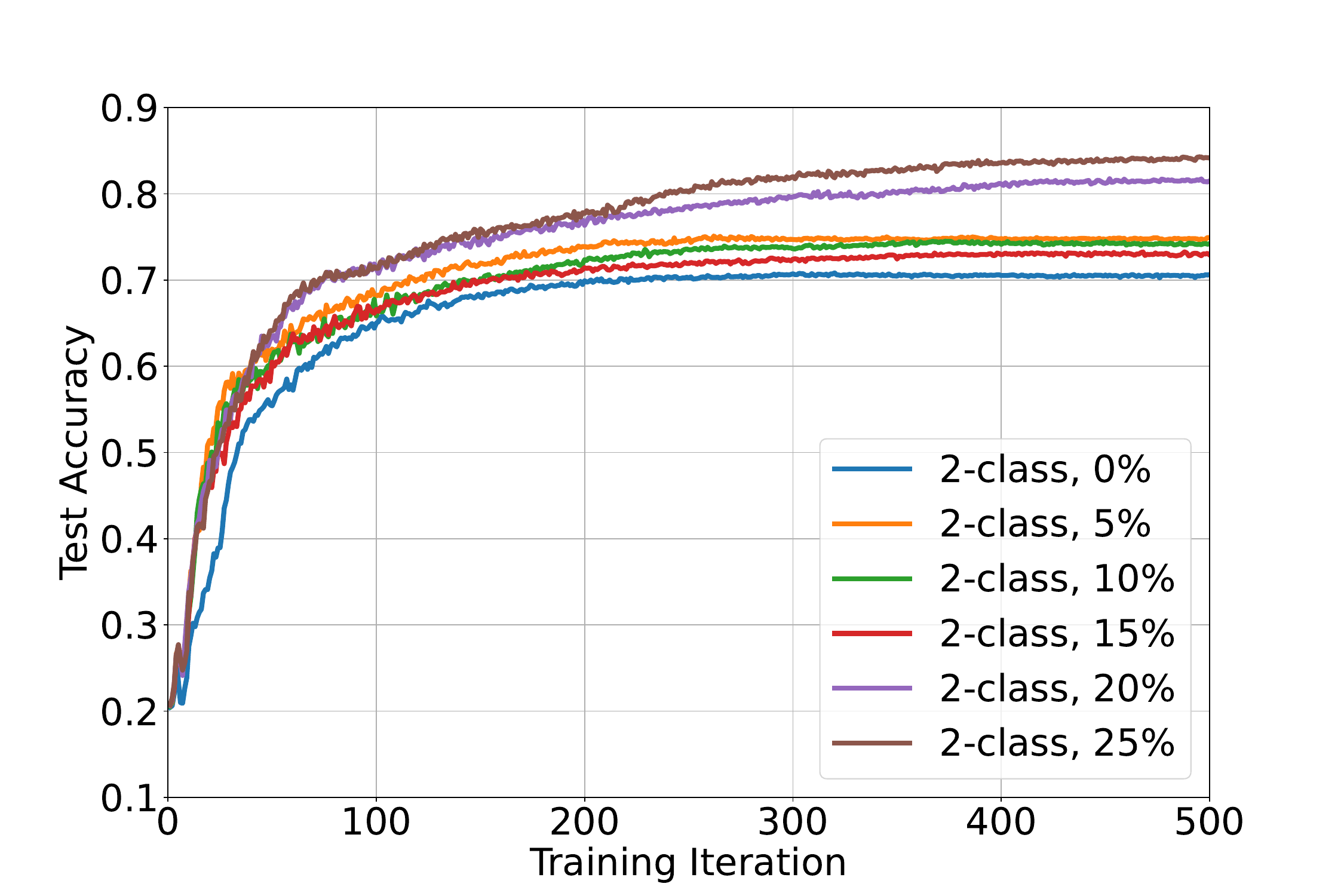}
\caption{FL model accuracy on \emph{CIFAR-10} dataset with varying amount of synthetic data samples when the edge servers are IID, each FL worker has two classes of data samples..}
\label{fig:cifardifferentpercent2class}
\end{figure}

In all three cases, adding synthetic data enhances the accuracy of the FL model on the \emph{MNIST} dataset. For instance, in Scenario~3, when no synthetic data is added ($0\%$), the FL model achieves an accuracy of $0.8923$ at the $500^{th}$ iteration. However, with just a $5\%$ addition of synthetic data, the accuracy improves to $0.9316$. The positive effect of synthetic data is particularly pronounced when each FL worker has only a single class of data, i.e., when the data distribution is more non-IID, as observed in Figures~\ref{fig:differentpercent1classedgeiid0} and~\ref{fig:differentpercent1classedgeiid1}. Similar trend is observed on the \emph{CIFAR-10} dataset. In Fig.~\ref{fig:cifardifferentpercent2class}, when the FL server has two classes of data samples and the edge servers have an IID distribution (Scenario~1), the FL model accuracy on the \emph{CIFAR-10} dataset increases from $0.7062$ when no synthetic data is added to $0.8416$ when $25\%$ of synthetic data is added at the $500^{th}$ iteration.

Additionally, we analyze how varying amounts of synthetic data affect FL model accuracy in Fig.~\ref{fig:diffferentpercentmnist}. Across all three scenarios, the most significant accuracy gain occurs with the introduction of $5\%$ synthetic data samples, especially when FL workers have only one class of data. However, adding larger amounts of synthetic data does not yield substantial further improvements. When $10\%$, $15\%$, $20\%$ or $25\%$ synthetic data samples are included, the accuracy curves nearly overlap, indicating diminishing returns. This indicates that even a modest amount of synthetic data can significantly improve the accuracy of the FL model.

\begin{figure}
\centering
\includegraphics[width=\linewidth]{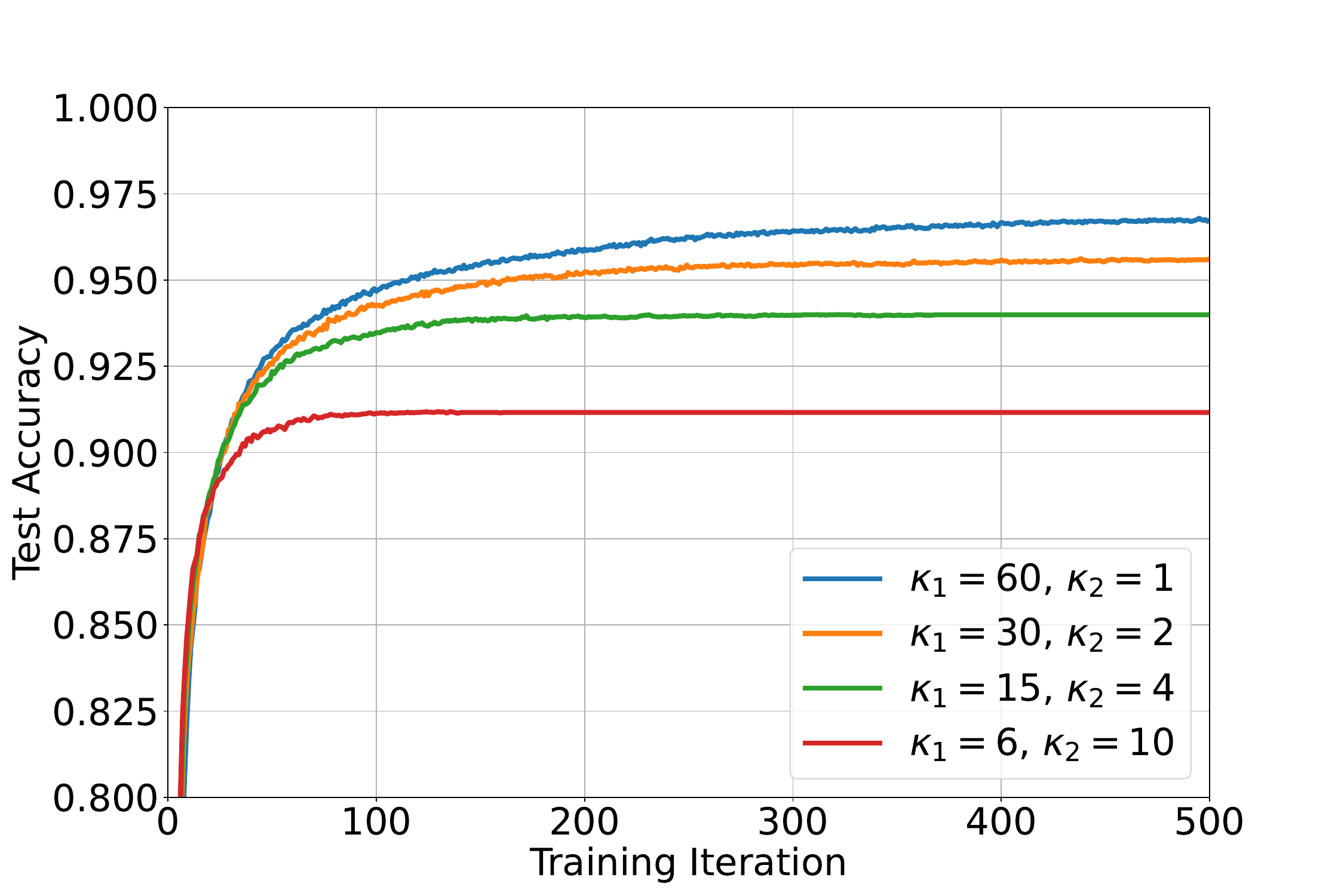}
\caption{FL model accuracy on \emph{MNIST} dataset given different values of $\kappa_1$ and $\kappa_2$ for fixed cloud interval, $\kappa_1\kappa_2$.}
\label{fig:diffkappa}
\end{figure}

\begin{figure}
\centering
\includegraphics[width=\linewidth]{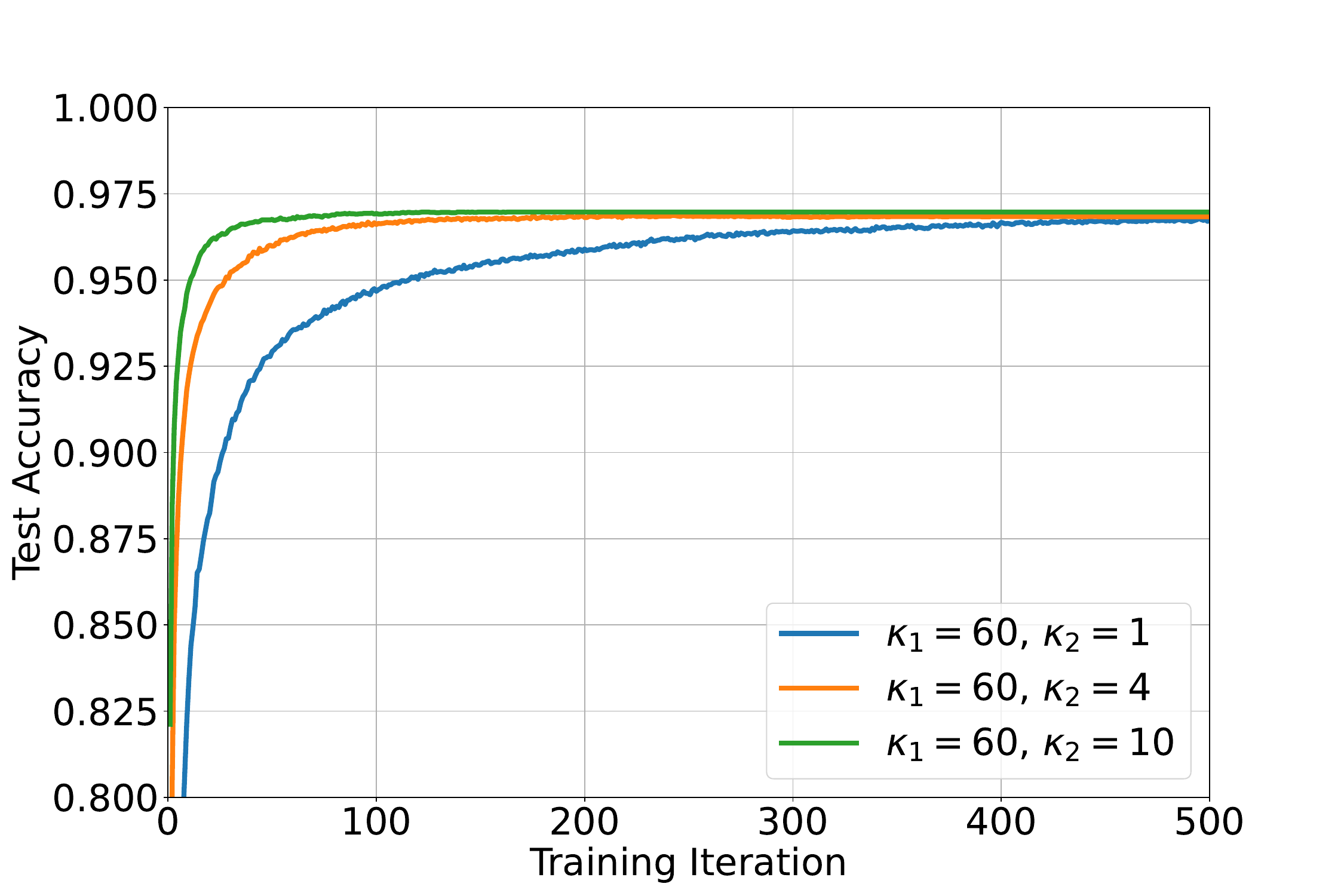}
\caption{FL model accuracy on \emph{MNIST} dataset given different values of $\kappa_2$ for fixed value of $\kappa_1$.}
\label{fig:diffkappa2}
\end{figure}

Furthermore, we examine the effect of varying numbers of local updates and edge aggregations on the accuracy of the FL model. In Fig.~\ref{fig:diffkappa}, we vary the number of local updates by the FL workers, $\kappa_1$, and the number of edge aggregations by the edge servers, $\kappa_2$, with the condition that the cloud aggregation takes place every $60$ iterations, i.e., $\kappa_1\kappa_2 = 60$ using $5\%$ synthetic data samples. We observe that as the number of local updates performed by the FL workers increases, the accuracy of the FL model increases, given that $\kappa_1\kappa_2$ remains constant. In other words, the target accuracy of the FL model can be achieved with fewer number of iterations when the FL workers perform more local computations on their own devices, given that the number of cloud aggregations, i.e., communication frequency with the FL server is fixed. In Fig.~\ref{fig:diffkappa2}, the numbers of local updates by FL workers, $\kappa_1$, is fixed at $60$. With a fixed total number of training iterations, $K=500$, as $\kappa_2$ increases, the length of each cloud interval $T$ increases and thus, the number of cloud aggregations decreases. We observe that as the value of $\kappa_2$ increases, the accuracy of the FL model also increases.

\section{Conclusion}
\label{sec:conclusion}

In this paper, we have proposed a synthetic-data-empowered HFL network to address the statistical challenges posed by non-IID data while motivating the participation of the FL workers. To encourage the participation of FL workers in the FL training process, we have adopted the evolutionary game to model the cluster formation of the edge servers. Then, we have analyzed the evolutionary game and performed numerical simulation to validate the performance of the network. The simulation results have shown that the using synthetic datasets generated by the edge servers allows higher FL accuracy, improving the performance of the FL model. Despite the need to allocate resources for training on the synthetic datasets, the evolutionary equilibrium is still achieved. In future works, we will consider a serverless approach to reduce the reliance on the FL server.

\bibliographystyle{ieeetr}
\bibliography{synthetic}

\end{document}